\documentclass{amsart}

\usepackage{booktabs} 
\usepackage{IEEEtrantools}
\usepackage{amssymb,amsfonts,amsmath}
\usepackage{subfigure}
\usepackage{graphicx}
\usepackage{paralist}
\usepackage{capt-of}
\usepackage{xcolor}
\usepackage{tikz}
\usetikzlibrary{calc,shapes,arrows}

\newtheorem{theorem}{Theorem}[section]
\newtheorem{lemma}[theorem]{Lemma}

\newtheorem{problem}[theorem]{Problem}

\newtheorem{proposition}[theorem]{Proposition}

\newtheorem{definition}[theorem]{Definition}
\newtheorem{assumption}[theorem]{Assumption}
\newtheorem{example}[theorem]{Example}

\newtheorem{remark}[theorem]{Remark}
\numberwithin{equation}{section}

\newcommand{\R}{{\mathbb{R}}}
\newcommand{\sys}{{\mathcal{S}}}

\newcommand{\B}{\mathcal{B}}
\newcommand{\I}{\mathcal{I}_d}

\newcommand{\N}{{\mathbb{N}}}

\newcommand{\eg}{{\it e.g.}}

\newcommand\norm[1]{\left\lVert#1\right\rVert}

\newcommand{\Let}{:=}

\newcommand{\intcc}[1]{\ensuremath{{\left[#1\right]}}}

\definecolor{myco}{rgb}{0.6, 0.73, 0.89}
\definecolor{myco1}{rgb}{0.0, 0.53, 0.74}

\begin{document}

\begin{abstract}
In this paper, we provide a compositional framework for synthesizing hybrid controllers for interconnected discrete-time control systems enforcing specifications expressed by co-B\"{u}chi automata. In particular, we first decompose the given specification to simpler reachability tasks based on automata representing the complements of original co-B\"{u}chi automata. Then, we provide a systematic approach to solve those simpler reachability tasks by computing corresponding control barrier functions. We show that such control barrier functions can be constructed compositionally by assuming some small-gain type conditions and composing so-called local control barrier functions computed for subsystems. We provide two systematic techniques to search for local control barrier functions for subsystems based on the sum-of-squares optimization program and counter-example guided inductive synthesis approach. Finally, we illustrate the effectiveness of our results through two large-scale case studies.         
\end{abstract}

\title[Compositional Construction of Control Barrier Functions]{Compositional Construction of Control Barrier Functions for Interconnected Control Systems$^\star$}
\author[P. Jagtap]{Pushpak Jagtap$^{1,\dagger}$}
\author[A. Swikir]{Abdalla Swikir$^{1,\dagger}$}
\author[M. Zamani]{Majid Zamani$^{2,3}$}
\address{$^1$Department of Electrical and Computer Engineering, Technical University of Munich, Germany.}
\email{\{pushpak.jagtap,abdalla.swikir\}@tum.de}
\address{$^2$Computer Science Department, University of Colorado Boulder, USA.}
\email{Majid.Zamani@colorado.edu}
\address{$^3$Computer Science Department, Ludwig Maximilian University of Munich, Germany.}
\thanks{$^\dagger$ The authors contributed equally to this work.}
\thanks{$^\star$ This work was supported in part by the H2020 ERC Starting Grant AutoCPS (grant agreement No. 804639), the German Research Foundation (DFG) through the grants ZA 873/1-1, and the TUM International Graduate School of Science and Engineering (IGSSE)}
\maketitle

\section{Introduction}
Formal synthesis of controllers for dynamical systems against complex logic specifications has gained considerable attentions in the last few years. These specifications are usually expressed using temporal logic formulae or (in)finite strings over automata. In the literature, the abstraction-based approaches are popular to solve such synthesis problems.  However, since the abstraction-based approaches usually require discretization of the state and input sets of concrete systems, the synthesis problem becomes very intractable for large-scale control systems.
To address this scalability issue, several results were proposed by utilizing the compositional abstraction-based synthesis where the synthesis is performed by computing the abstractions and (possibly) controllers for smaller subsystems; see the results in \cite[and references therein]{meyer,pola7,SWIKIR2019,8728138} for more details. 

Alternatively, a discretization-free approach, based on control barrier functions, has shown a potential to solve the formal synthesis problems as well. Assuming a prior knowledge of control barrier functions, several techniques have been recently introduced to ensure the safety of dynamical systems (see \cite[ and the references
therein]{ames2016control,ames2019control}), or the satisfaction of a set of signal temporal logic tasks for multi-agent systems \cite{Dimarogonas,Dimarogonas2}.
The results in \cite{jagtap2019formal} provide techniques to search for parametric control barrier functions to synthesize controllers for stochastic control systems enforcing a class of temporal logic specifications over finite time horizons. 
Though promising, the computational complexity of searching for parametric control barrier functions grows in polynomial time \cite{jagtap2018temporal,wongpiromsarn2015automata} with respect to the dimension of the system and, hence, the existing approaches \cite{ames2016control,ames2019control,jagtap2019formal} will also become computationally intractable while dealing with large-scale interconnected systems.

Motivated by the above results and their limitations, this work proposes a controller synthesis approach for large-scale systems against complex logic specifications via compositional construction of control barrier functions.  
To the best of our knowledge, this paper is the first to utilize compositional construction of control barrier functions for synthesizing hybrid controllers for interconnected discrete-time control systems against specifications expressed by co-B\"{u}chi automata. In order to achieve this, we first decompose the given specification to simpler reachability tasks based on automata representing the complements of original co-B\"{u}chi automata. Then, we provide a systematic approach to solve those simpler tasks by computing corresponding control barrier functions. Those control barrier functions are obtained by composing so-called local control barrier functions while utilizing some small-gain type conditions.  
In the final step, we combine those control barrier functions and controllers obtained for simpler tasks to obtain hybrid controllers ensuring the desired complex specifications over large-scale interconnected systems. 
In addition, we provide two systematic approaches to search for parametric local control barrier functions under suitable assumptions on the dynamics of the subsystems. The first approach is using the sum-of-square optimization \cite{parrilo2003semidefinite} and the second one is utilizing a counter-example guided inductive synthesis approach \cite{ravanbakhsh}. 

Finally, we demonstrate the effectiveness of the proposed results on two large-scale case studies with $10^4$ state dimensions. First, we apply our results to the temperature regulation in a circular building by synthesizing controllers for a network containing $N$ rooms for any $N\geq3$ ensuring the satisfaction of a specification given by a deterministic co-B\"{u}chi automaton. Additionally, we also apply the proposed techniques to a nonlinear example of a fully connected network of Kuramoto oscillators and synthesize hybrid controllers ensuring the satisfaction of a given specification.

\section{Notation and Preliminaries}\label{1:II}
\subsection{Notation}
We denote by $\R$ and $\N$ the set of real numbers and non-negative integers,  respectively.
These symbols are annotated with subscripts to restrict them in
an obvious way, \eg, $\R_{>0}$ denotes the positive real numbers. We denote the closed, open, and half-open intervals in $\R$ by $[a,b]$,
$(a,b)$, $[a,b)$, and $(a,b]$, respectively. For $a,b\in\N$ and $a\le b$, we
use $[a;b]$, $(a;b)$, $[a;b)$, and $(a;b]$ to
denote the corresponding intervals in $\N$.
Given $N\in\N_{\ge1}$, vectors $\nu_i\in\R^{n_i}$, $n_i\in\N_{\ge1}$, and $i\in[1;N]$, we
use $\nu=[\nu_1;\ldots;\nu_N]$ to denote the vector in $\R^n$ with
$n=\sum_i n_i$ consisting of the concatenation of vectors~$\nu_i$. 
Note that given any  $\nu\in\R^{n}$, $\nu \ge 0$ if $\nu_i \ge 0$ for any $i \in [1;n]$.
We use $\textbf{1}_n$ to denote a vector in $\R^n$ with all elements being one.
The individual elements in a matrix $A\in \R^{m\times n}$ are denoted by $\{A\}_{ij}$, where  $i\in\intcc{1;m}$ and $j\in\intcc{1;n}$. We use $\norm{\cdot}$ to denote the infinity norm. Given any $a\in\R$, $\vert a\vert$ denotes the absolute value of $a$. Given sets $X$ and $Y$, we denote by $f:X\rightarrow Y$ an ordinary map from $X$ to $Y$.

We denote the empty set by $\emptyset$. 
Given a set $S$, the notation $|S|$ denotes the cardinality of $S$ and $S^*$ and $S^\omega$ denote the set of all finite and infinite strings over $S$, respectively. 
Given sets $U$ and $S\subset U$, the complement of $S$ with respect to $U$ is defined as $U\backslash S = \{x : x \in U, x \notin S\}.$
We use notations $\mathcal{K}$ and $\mathcal{K}_\infty$
to denote different classes of comparison functions, as follows:
$\mathcal{K}=\{\alpha:\mathbb{R}_{\geq 0} \rightarrow \mathbb{R}_{\geq 0} |$ $ \alpha$ is continuous, strictly increasing, and $\alpha(0)=0\}$; $\mathcal{K}_\infty=\{\alpha \in \mathcal{K} |$ $ \lim\limits_{r \rightarrow \infty} \alpha(r)=\infty\}$.
For $\alpha,\gamma \in \mathcal{K}_{\infty}$ we write $\alpha<\gamma$ if $\alpha(s)<\gamma(s)$ for all $s>0$. Function $\mathcal{I}_d\in\mathcal{K}_{\infty}$ denotes the identity one. We use notations $\top$ and $\bot$ to represent \texttt{true} and \texttt{false}, respectively.

\subsection{Interconnected Control Systems}\label{Interconnect_sys1}
First, we define discrete-time control subsystems which will be later interconnected to form a large-scale discrete-time control system.  
\begin{definition}  
	A control subsystem $\sys_i$ is a tuple
	\begin{align}
	\label{eq:controlsystem}
	\sys_i = (X_i,U_i,W_i,f_i,Y_i,h_i), \quad i \in [1;N],
	\end{align}
	where $X_i$, $U_i$, $W_i$, and $Y_i$ are the state set, the external input set, the internal input set, and the output set, respectively. 
	The function $ f_i: X_i \times U_i \times W_i \rightarrow X_i$ is the transition function and $h_i : X_i \rightarrow Y_i $ is the output function.	
	The discrete-time control subsystem $\sys_i $ is described by difference equations of the form
	\begin{align}\label{eq:two}
	\sys_i:\left\{
	\begin{array}{rl}
	\mathbf{x}_i(k+1)=&\!\!\!\! f_i(\mathbf{x}_i(k),\nu_i(k),\omega_i(k)),\\
	\mathbf{y}_i(k)=&\!\!\!\!h_i(\mathbf{x}_i(k)),
	\end{array}
	\right.
	\end{align}
	where $\mathbf{x}_i:\N\rightarrow X_i $, $\mathbf{y}_i:\N\rightarrow Y_i$, $\nu_i:\N\rightarrow U_i$, and $\omega_i:\N\rightarrow W_i$ are the state run, output run, external input run, and internal input run, respectively. 
\end{definition}  	
Now, we provide a formal definition of interconnected discrete-time control systems.
\begin{definition}
	\label{interconnectedsystem} 
	Consider $N \in \N_{\geq 1}$ control subsystems $\sys_i = (X_i,U_i,W_i,f_i,Y_i,h_i)$ with their inputs and outputs partitioned as
	\begin{align*}	
	w_i &\!=\!\! [w_{i1};\dots;w_{i(i-1)};w_{i(i+1)};\dots;w_{iN}],  W_i\!=\!\prod_{j=1,j\neq i}^{N} \!\!W_{ij},\\
	y_i &= [y_{i1};\dots;y_{iN}], Y_i=\prod_{j=1}^N  Y_{ij},
	\end{align*}
	with $w_{ij} \in W_{ij}$, $y_{ij} =h_{ij}(x_i) $ and output function
	\begin{align*}
	h_i(x_i) = [h_{i1}(x_i);\dots;h_{iN}(x_i)]~~ \text{with}~~ h_{ii}(x_i)=x_i.
	\end{align*}
	
	The interconnected control system $\sys=\mathcal{I}(\sys_1,\dots,\sys_N)$ is a tuple
	\begin{align}
	\label{interConnectedsys}
	\sys = (X,U,f),
	\end{align} described by the difference equation
	\begin{align}\label{int}
	\mathbf{x}(k+1)= f(\mathbf{x}(k),\nu(k)),
	\end{align}
	where $X=\prod_{i=1}^N X_i$, $U=\prod_{i=1}^N U_i$, and function    
	\begin{align}\notag
	f(x,u)&=[f_1(x_1,u_1,w_1);\dots;f_N(x_N,u_N,w_N)],
	\end{align} 
	where $x = [x_1;\dots;x_N]\in X$, $u = [u_1;\dots;u_N]\in U$, and the interconnection variables are constrained by $w_{ij} = y_{ji}$, $Y_{ji}\subseteq W_{ij}$, $\forall i,j \in [1;N], i\neq j$. Moreover, let  $\mathbf{x}_{x,\nu}$ denote a state run of $\sys$ starting from initial state $x\in X$ under
	input run $\nu:\N\rightarrow U$.
	An example of the interconnection of three control subsystems $\sys_1$, $\sys_2$,
	and $\sys_3$ is illustrated in Figure \ref{system1}.
\end{definition}
In the above definition, we assumed that one has access to the full state information of subsystems (i.e. $h_{ii}(x_i)=x_i$) for the sake of controller synthesis. However, for the sake of internal interconnections, we work with the outputs of states (i.e. $h_{ij},i,j\in[1;N],i\neq j$) (cf. Figure \ref{system1}). 
\begin{figure}
	\centering
	\includegraphics[scale=0.16]{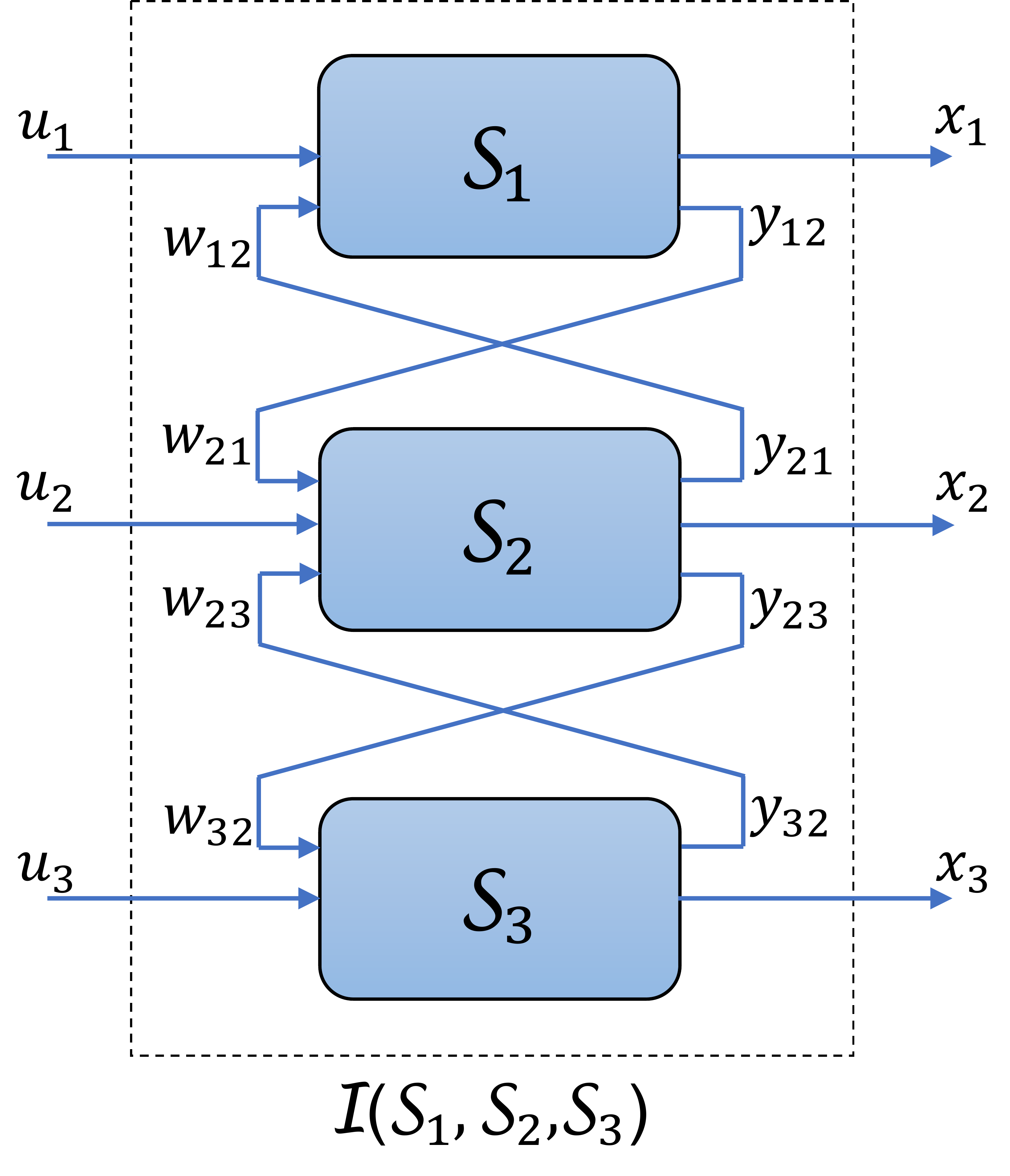}
	\caption{Interconnection of three control subsystems $\sys_1$, $\sys_2$, and $\sys_3$ with $h_{13}$ and $h_{31}$ being zero maps.}
	\label{system1}
	\vspace{-1em}
\end{figure}

We are interested in synthesizing control policies $\rho$ for system $\sys$ enforcing given complex specifications. Here, we consider \textit{history-dependent policies} given by $\rho=(\rho_0,\rho_1,\ldots,\rho_k,\ldots)$ with functions $\rho_k: \mathcal{H}_k\rightarrow U$, where $\mathcal{H}_k$ is the set of all $k$-histories $\mathbf h_k$ defined as $\mathbf h_k:=(\mathbf x(0),\nu(0),\mathbf x(1),\nu(1),\ldots,\mathbf x(k-1),\nu(k-1),\mathbf x(k))$. A subclass of those policies are called \emph{stationary} and are defined as $\rho=(\mathbf u,\mathbf u,\ldots,\mathbf u,\ldots)$ with a function $\mathbf u: X\rightarrow U$. In stationary policies, the mapping at time $k$ depends only on the current state $\mathbf x(k)$ and does not change over time.

\subsection{Class of Specifications}
\label{weak}
Here, we consider the class of specifications expressed by deterministic co-B{\"u}chi automata (DCA) \cite{LODI} as defined next.
\begin{definition}\label{DBA}
	A deterministic co-B{\"u}chi automaton (DCA) is a tuple $\mathcal{A}=(Q,Q_0,\Sigma,\delta,F)$, where $Q$ is a finite set of states, $Q_0\subseteq Q$ is a set of initial states, $\Sigma$ is a finite set of alphabet, $\delta: Q\times\Sigma\rightarrow Q$ is a transition function, and $F\subseteq Q$ is a set of final states.
\end{definition}
We use notation $q\overset{\sigma}{\longrightarrow} q'$ to denote transition $(q,\sigma,q')\in\delta$. We also denote the set of all successor states of a state $q\in Q$ by $\Delta(q)$. Consider an infinite state run $\textbf{q}=(q_0,q_1,\ldots)\in Q^{\omega}$ such that $q_0\in Q_0$, $q_i \overset{\sigma_i}{\longrightarrow} q_{i+1}$ for all $i\geq 0$ and let $\mathsf{Inf}(\textbf q)$ be the set of states that occurs infinitely many times in $\textbf{q}$.
An infinite word (a.k.a trace) $\sigma=(\sigma_0,\sigma_1,\ldots)\in \Sigma^{\omega}$ is accepted by DCA $\mathcal{A}$ if there exists an infinite state run $\textbf{q}$ such that $\mathsf{Inf}(\textbf q)\cap F=\emptyset$. The set of words accepted by $\mathcal{A}$ is called the accepting language of $\mathcal{A}$ and is denoted by $\mathcal{L}(\mathcal{A})$.

A deterministic B{\"u}chi automaton (DBA) is defined syntactically exactly as a deterministic co-B{\"u}chi automaton except that
its accepting runs are those for which $\mathsf{Inf}(\textbf q)\cap F\ne\emptyset$. Note that the complement of a deterministic co-B{\"u}chi automaton is a deterministic B{\"u}chi automaton \cite{LODI}.

In this work, we consider those specifications given by the accepting languages of DCA $\mathcal{A}$ defined over the set of atomic propositions $\Pi$, i.e., the alphabet\footnote{For properties expressed by DCA $\mathcal{A}$ over atomic propositions $\Pi$, $\mathcal{A}$ is usually
	constructed over the alphabet $\Sigma=2^{\Pi}$. Without loss of generality, we work with the set $\Pi$ directly as the alphabet rather than its power set.} $\Sigma=\Pi$.  We should highlight that the temporal logic specifications represented using \emph{obligation} properties \cite{manna2012temporal} (including boolean combinations of safety and guarantee properties) are all recognized by deterministic weak automata \cite{dax2007mechanizing} which are included in DCA. For other temporal logic formulae, one can readily check the existence of DCA using the tool \texttt{SPOT} \cite{duret2016spot}. 

\subsection{Satisfaction of Specifications by Interconnected Control Systems}
In this subsection, we define how the specification given by the accepting language of DCA $\mathcal{A}$ is satisfied by the system $\sys$ as in Definition \ref{interconnectedsystem}. To relate the state of the system to DCA $\mathcal{A}$, we use a measurable labeling function $L: X \rightarrow \Pi$, where $\Pi$ is the set of atomic propositions.
\begin{definition}\label{sys_trace1}
	Consider an interconnected control system $\sys = (X,U,f)$ as in Definition \ref{interconnectedsystem} and a specification expressed by DCA $\mathcal{A}=(Q,Q_0,\Pi,\delta,F)$. In order to reason about the given specification for the system $\sys$, we use a measurable labeling function $L: X \rightarrow \Pi$. In addition, consider an infinite state run $\mathbf{x}=(\mathbf x(0),\mathbf x(1),$ $\ldots) \in X^\omega$, and labeling function $L: X \rightarrow \Pi$. Then, the corresponding trace over $\Pi$ is given by $L(\mathbf{x}):=(\sigma_0,\sigma_1,\ldots) \in\Pi^\omega$, where $\sigma_k=L(\mathbf x(k))$ for all $k \in\{0,1,\ldots\}$.
\end{definition}
Note that we abuse the notation by using map $L(\cdot)$ over $X^\omega$, i.e., $L(\mathbf x(0),\mathbf x(1),\ldots)\equiv (L(\mathbf x(0)),L(\mathbf x(1)),\ldots)$.
Their distinction is clear from the context. Next we define the satisfaction of specifications by the control systems $\sys$.
\begin{definition}
	Consider an interconnected control system $\sys = (X,U,f)$ as in Definition \ref{interconnectedsystem}, a specification given by the accepting language of DCA $\mathcal{A}=(Q,Q_0,\Pi,\delta,F)$, and a labeling function $L: X \rightarrow \Pi$. We say that the state run of $\sys$ starting from initial state $x\in X$ under control policy $\rho$ satisfies the specification given by $\mathcal A$, denoted by $L(\mathbf{x}_{x,\rho})\models \mathcal A$, if $L(\mathbf{x}_{x,\rho})\in\mathcal{L}(\mathcal A)$. 
\end{definition}

\subsection{Problem Definition}
The main synthesis problem in this work is formally defined next. 
\begin{problem}\label{prob1}
	Given an interconnected control system $\sys\!=\! (X,U,f)$ as in Definition \ref{interconnectedsystem}, a specification given by the accepting language of DCA $\mathcal{A}=(Q,Q_0,\Pi,\delta,F)$ over a set of atomic propositions $\Pi=\{p_0,p_1,\ldots,$ $p_M\}$, and a labeling function $L: X \rightarrow \Pi$, compute a control policy $\rho$ (if existing) such that $L(\mathbf{x}_{x,\rho})\models \mathcal A$ for all $x\in L^{-1}(p_i)$ and some $i\in\{0,1,\ldots,M\}$.
\end{problem}
Finding a solution to Problem \ref{prob1} (if existing) is difficult in general. In this paper, we provide a method that is sound in solving the problem. To construct a control policy $\rho$, our approach utilizes the notion of control barrier functions as defined in the next section. Later, we provide a compositional approach on constructing such control barrier functions to make it tractable for large-scale systems.
\section{ Control Barrier Function}
In this section, we define the notion of control barrier function which will be used throughout the paper.
\begin{definition}\label{bc}
	A function $ \B:X \to \R_{\geq0}$ is a control barrier function for an interconnected control system $\sys=(X,U,f)$ as in Definition \ref{interconnectedsystem} if for any $ x\in X$ there exists an input $u\in U$ such that
	\begin{align}\label{bc3}
	\B(f(x,u))&\leq\kappa(\B(x)),	
	\end{align}
	for some $\kappa\in \mathcal{K}_{\infty}$ with $\kappa\leq \I$.
\end{definition}
Note that the above definition associates a stationary policy $\mathbf u: X\rightarrow U$ according to the existential quantifier on the input for any state $x \in X$. The importance of the existence of a control barrier function for system $\sys$ is shown in the following proposition.
\begin{proposition}\label{thm1}
	Consider an interconnected control system $\sys=(X,U,f)$, and sets $X_a,X_{b}\subseteq X$. Assume that there exits a control barrier function $ \B:X \to \R_{\geq0}$ as defined in Definition \ref{bc} with a stationary policy $\mathbf u:X\rightarrow U$ and constants $\epsilon_1,\epsilon_2\in\R_{>0}$ with $\epsilon_2\geq\epsilon_1$ such that 
	\begin{align}
	\B(x)&\leq \epsilon_1, \quad\quad\forall x\in X_a, \label{bc1}\\
	\B(x)&> \epsilon_2, \quad\quad\forall x\in X_{b}. \label{bc2}
	\end{align}
	Then, for the state run $\mathbf{x}_{x,\mathbf u}$ of $\sys$ starting from any initial state $x \in X_a$ and under corresponding policy $\mathbf u(\cdot)$, one has $\mathbf{x}_{x,\mathbf u}(k)\cap X_{b}\!=\!\emptyset$, $\forall k\in \N$. 	
\end{proposition}
\begin{proof}
	We prove by contradiction. Consider a state run $\mathbf{x}_{x,\mathbf u}$ of $\sys$ that starts at some $x \in X_a$. Suppose $\mathbf{x}_{x,\mathbf u}$ reaches a state inside $X_b$. Following \eqref{bc1} and \eqref{bc2}, one has $\B(\mathbf x(0))\leq\epsilon_1$ and $\B(\mathbf x(k))>\epsilon_2$ for some $k\in\N$. Since $\B(\cdot)$ is a control barrier function and by using inequality \eqref{bc3}, one can conclude that $\epsilon_2<\B(\mathbf x(k))\leq \B(\mathbf x(0)) \leq\epsilon_1$. This contradicts  $\epsilon_2\geq\epsilon_1$ which completes the proof.
	\begin{figure}
		\begin{center}
			\includegraphics[height=3.4cm]{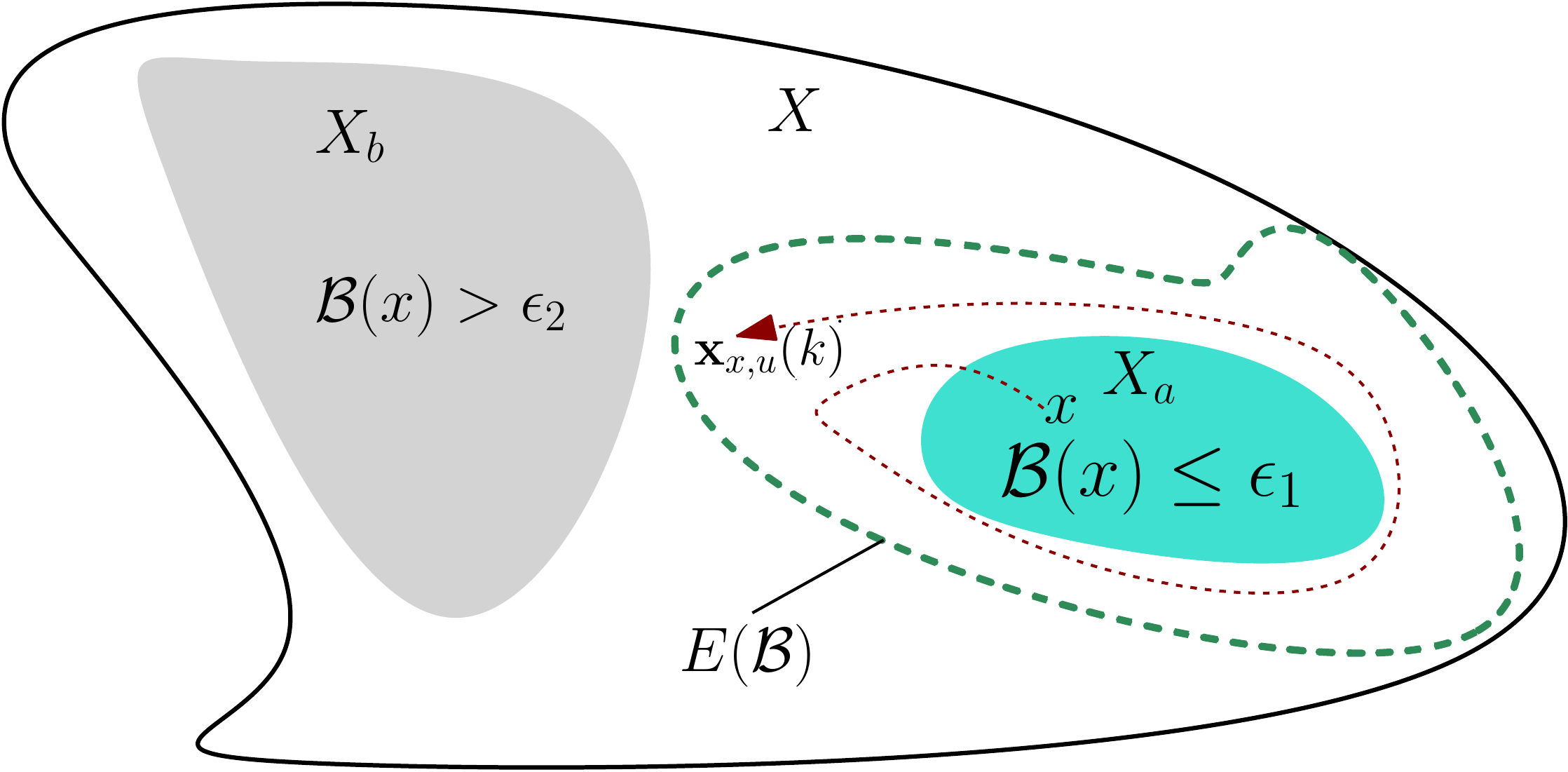}
			\caption{Illustration of a set $X$ containing sets $X_a$ and $X_{b}$: the dashed line
				illustrates the $\epsilon$-level set of $\B$, defined as $E(\B)\!\!=\!\{x\!\in \!X|\B(x)\!=\!\epsilon\}$, and the dotted curve is the run of system $\sys$.}
			\label{bf}
		\end{center}
		\vspace{-1.2em}
	\end{figure}
\end{proof}
The interpretation of Proposition \ref{thm1} is illustrated in Figure \ref{bf}. In the next section, we discuss how to translate Problem \ref{prob1} for a given specification into the computation of a collection of control barrier functions each
satisfying conditions as in Proposition \ref{thm1}.

\section{Formal Synthesis using Control Barrier Functions}
In order to synthesize control policies using control barrier functions enforcing specifications expressed by DCA $\mathcal A$, we first provide the decomposition of specifications into sequential reachability tasks which will later be solved using control barrier functions.
\subsection{Sequential Reachability Decomposition}
\label{compo_runs}
Consider a DCA $\mathcal{A}=(Q,Q_0,\Pi,\delta,F)$ expressing the properties of interest for the system $\sys$. Consider the DBA $\mathcal{A}^c=(Q,Q_0,\Pi,\delta,F)$ whose language is the complement of the language of DCA $\mathcal A$. As one can readily see, the DBA $\mathcal{A}^c$ has the same structure as the DCA $\mathcal{A}$, but with the B\"uchi accepting condition. The infinite sequence $\textbf{q}=(q_0,q_1,\ldots)\in Q^{\omega}$ is called an accepting state run if $q_0\in Q_0$ and there exists infinitely many $j \geq 0$ such that $q_j\in F $, and there exists an infinite word $\sigma = (\sigma_0,\sigma_1,\ldots)\in\Pi^\omega$ such that $q_k \overset{\sigma_k}{\longrightarrow} q_{k+1}$ for all $k\in\mathbb{N}$. For a given accepting state run $\textbf{q}$, we denote the corresponding infinite words by $\sigma(\textbf{q})\subseteq \Pi^\omega$. We also use a similar notation to denote finite words corresponding to finite state runs (i.e., $\sigma(\overline{\textbf{q}})\in\Pi^{n}$ for $\overline{\textbf{q}}\in Q^{n+1}, n\in\N$). It is known \cite[Lemma 4.39]{baier2008principles} that there exists a word $\sigma\in \Pi^\omega$ accepted by $\mathcal{A}^c$ if and only if there exists a state run of $\mathcal{A}^c$ of the form $\textbf{q}= (q_0^r,q_1^r,\ldots,q_{\mathsf m_r}^r,(q_0^s,q_1^s,\ldots,q_{\mathsf m_s}^s)^{\omega})\in Q^\omega$, where $\mathsf m_r, \mathsf m_s\in\N $ with $\mathsf m_r+ \mathsf m_s=n $, $q_0^r\in Q_0$ and $q_0^s\in F$. Let $\overline{\textbf{q}}$ be a finite state run fragment of an accepting run $\textbf{q}$ constructed by considering infinite sequence $(q_0^s,q_1^s,\ldots,q_{\mathsf m_s}^s)$ only once and is given by $\overline{\mathbf{q}}= (q_0^r,q_1^r,\ldots,q_{\mathsf m_r}^r,q_0^s,q_1^s,\ldots,q_{\mathsf m_s}^s,q_0^s)\in Q^*$. \\
Let $\mathcal{R}$ be the set of all such finite state run fragments excluding self-loops,
\begin{align}
\label{eq_compo_runs}
\mathcal{R}:= \{\overline{\textbf{q}} =& (q_0^r,q_1^r,\ldots,q_{\mathsf m_r}^r,q_0^s,q_1^s,\ldots,q_{\mathsf m_s}^s,q_0^s)\mid q_0^r\in Q_0, q_0^s \in F, q^r_i \neq  q^r_{i+1},\forall i< {\mathsf m_r},\text{ and }q^s_j \neq  q^s_{j+1},\forall j\hspace{-.2em} <\hspace{-.2em} {\mathsf m_s}\}.
\end{align}
Computation of $\mathcal{R}$ can be done algorithmically by viewing $\mathcal{A}^c$ as a directed graph $\mathcal{G}=(\mathcal{V},\mathcal{E})$ with vertices $\mathcal{V}=Q$ and edges $\mathcal{E}\subseteq\mathcal{V}\times\mathcal{V}$ such that $(q,q')\in\mathcal{E}$ if and only if $q'\neq q$ and there exist $p\in\Pi$ such that $q\overset{p}{\longrightarrow} q'$. We call a finite sequence of states $(q_0,q_1,\ldots,q_n)\in Q^n,n\in\mathbb{N}$, satisfying $(q_i,q_{i+1})\in\mathcal{E}$, for all $i\in[0;n-1]$ a path in the graph $\mathcal{G}$. For any $(q,q')\in\mathcal{E}$, we denote the atomic proposition associated with the edge $(q,q')$ by $\sigma(q,q')$. 
Now, one can easily compute $\mathcal{R}$ using variants of depth first search algorithm \cite{russell2003artificial} over $\mathcal{G}$.
For each $p \in \Pi$, we define a set $\mathcal{R}^p$ as
\begin{equation}
\label{eq_compo_runs1}
\mathcal{R}^p \hspace{-.2em}:= \hspace{-.2em}\{\overline{\textbf{q}} =(q_0^r,q_1^r,\ldots,q_{\mathsf m_r}^r,q_0^s,q_1^s,\ldots,q_{\mathsf m_s}^s,q_0^s)\hspace{-.2em}\in\hspace{-.2em}\mathcal R \hspace{-.2em}\mid\hspace{-.2em} \sigma(q^r_0,q^r_1)=p\}.
\end{equation}
Decomposition into sequential reachability is performed as follows.
For any $\overline{\textbf{q}} =(q_0,q_1,\ldots,q_{\mathsf m_r+\mathsf m_s+3})\in\mathcal{R}^p$, we define $\mathcal{P}^p(\overline{\textbf{q}})$ as a set of all state runs of length $3$,
\begin{equation}
\label{eq_compo_reach}
\mathcal{P}^p(\overline{\textbf{q}}):=\{\left(q_i,q_{i+1},q_{i+2},\right)\mid 0\leq i\leq \mathsf m_r+\mathsf m_s+1\}.
\end{equation}
We define $\mathcal{P}(\mathcal{A}^c)=\bigcup_{p\in\Pi}\bigcup_{\overline{\textbf{q}}\in\mathcal{R}^p}\mathcal{P}^p(\overline{\textbf{q}})$. 
For the better understanding, the decomposition into sequential reachability is demonstrated below with an example.
\def\example{\par\noindent{\bf Example 1.} \ignorespaces}
\def\endexaple{}

\begin{example}
	\begin{figure}
		\centering
		\includegraphics[scale=0.2]{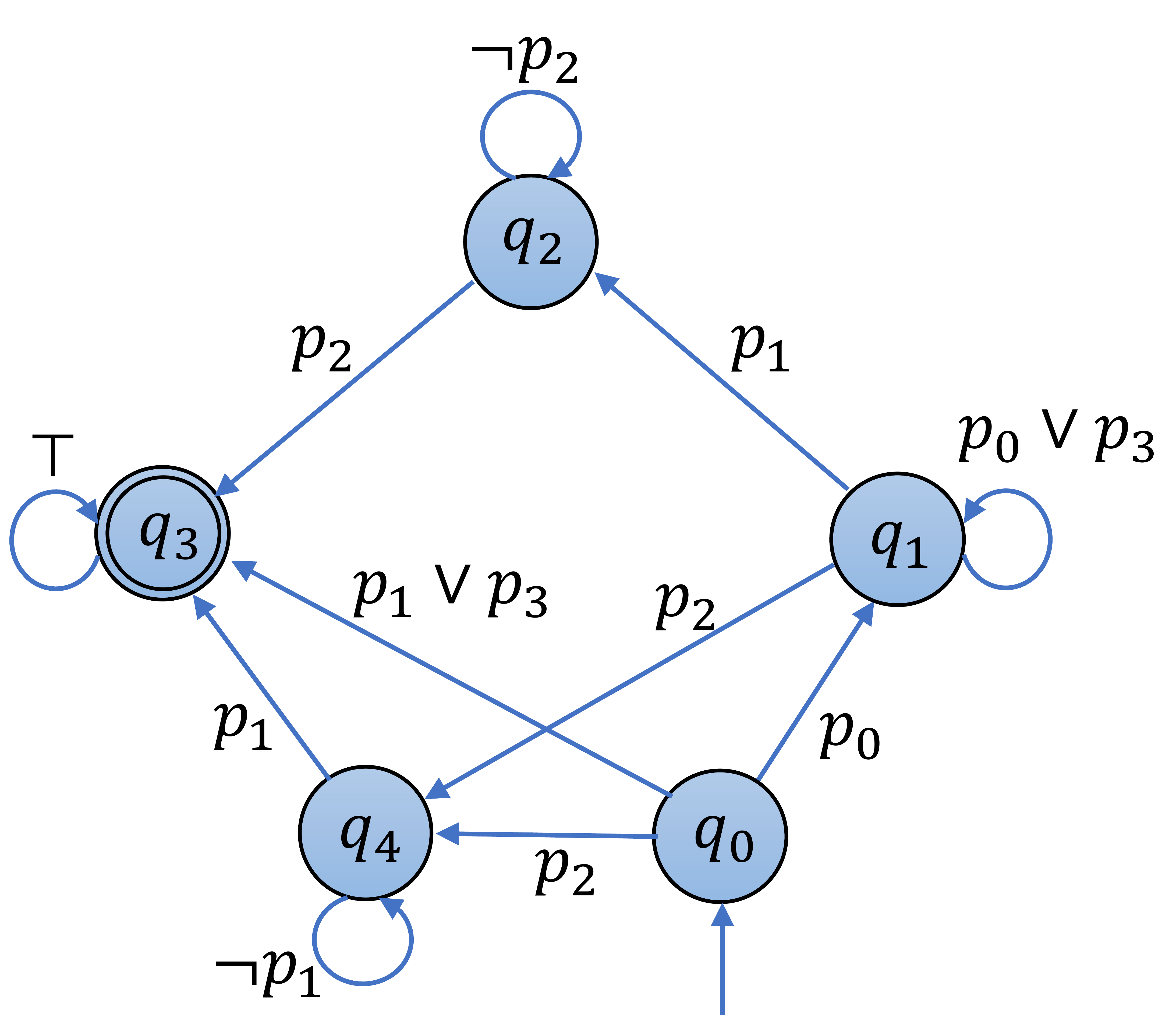}
		\caption{DBA $\mathcal{A}^c$ used in Example 1.}
		\label{fig_compo_automata}
	\end{figure}
	Consider a DBA $\mathcal{A}^c$ as shown in Figure \ref{fig_compo_automata}, where $Q_0=\{q_0\}$, $\Pi=\{p_0,p_1,p_2,p_3\}$, and $F=\{q_3\}$. 
	The set of accepting state runs is $\{(q_0,q_1^*,q_2^*,q_3^\omega), (q_0,q_1^*,q_4^*,q_3^\omega), (q_0,q_4^*,q_3^\omega), (q_0,q_3^\omega)\}$. The set of finite state run fragments $\mathcal{R}$ in \eqref{eq_compo_runs} is obtained as follow:
	\begin{equation*}
	\mathcal{R}=\{(q_0,q_1,q_2,q_3,q_3),(q_0,q_1,q_4,q_3,q_3),(q_0,q_4,q_3,q_3),(q_0,q_3,q_3)\}.
	\end{equation*}
	The sets $\mathcal{R}^p$ for $p\in\Pi$ are as follows:
	\begin{align*}
	&\mathcal{R}^{p_0}=\{(q_0,q_1,q_2,q_3,q_3),(q_0,q_1,q_4,q_3,q_3)\},\quad\mathcal{R}^{p_1}=\{(q_0,q_3,q_3)\}, \\ &\mathcal{R}^{p_2}=\{(q_0,q_4,q_3,q_3)\},\quad \mathcal{R}^{p_3}=\{(q_0,q_3,q_3)\}.
	\end{align*} 
	The sets $\mathcal{P}^p(\overline{\textbf{q}})$ for $\overline{\textbf{q}}\in\mathcal{R}^p$ are as follows:
	\begin{align*}
	&\mathcal{P}^{p_0}(q_0,q_1,q_2,q_3,q_3)=\{(q_0,q_1,q_2),(q_1,q_2,q_3),(q_2,q_3,q_3)\},\\
	&\mathcal{P}^{p_0}(q_0,q_1,q_4,q_3,q_3)=\{(q_0,q_1,q_4),(q_1,q_4,q_3),(q_4,q_3,q_3)\},\\
	&\mathcal{P}^{p_2}(q_0,q_4,q_3,q_3)=\{(q_0,q_4,q_3),(q_4,q_3,q_3)\},\\
	&\mathcal{P}^{p_1}(q_0,q_3,q_3)=\mathcal{P}^{p_3}(q_0,q_3,q_3)=\{(q_0,q_3,q_3)\}.
	\end{align*} 
	For every $\overline{\textbf{q}}\in\mathcal{R}^p$, the corresponding finite words $\sigma(\overline{\textbf{q}})$ are listed as follows:
	\begin{align*}
	& \sigma(q_0,q_3,q_3)=(p_1\vee p_3,\top), \sigma(q_0,q_1,q_2,q_3,q_3)=(p_0,p_1,p_2,\top),\\
	& \sigma(q_0,q_4,q_3,q_3)=(p_2,p_1,\top), \sigma(q_0,q_1,q_4,q_3,q_3)=(p_0,p_2,p_1,\top).
	\end{align*}
\end{example}

Having $\mathcal{P}^p(\overline{\textbf{q}})$ defined in \eqref{eq_compo_reach} as the set of state runs of length $3$, now we provide a systematic approach to compute a policy such that the state runs of $\sys$ satisfy the specification expressed by DCA $\mathcal A$.  
Given DBA $\mathcal A^c$, our approach relies on performing computation of control barrier functions for each element of $\mathcal P(\mathcal{A}^c)$, which at the end provides control policies ensuring that we never have accepting runs in the complement of the given specification (i.e., DCA $\mathcal A$). To provide the result on the construction of control policies to solve Problem \ref{prob1}, we provide the following lemma which is a direct consequence of results in Proposition \ref{thm1} and, hence, provided without a proof.
\begin{lemma}\label{compo_lemma1}
	For $p\in\Pi$ and $\overline{\textbf{q}}\in\mathcal{R}^p$, consider $(q,q',q'')\in\mathcal P^p(\overline{\textbf{q}})$. If there exists a control barrier function with stationary policy $\mathbf u(\cdot)$ satisfying conditions \eqref{bc1} and \eqref{bc2} in Proposition \ref{thm1} with $X_a=L^{-1}(\sigma(q,q'))$ and $X_{b}=L^{-1}(\sigma(q',q''))$, then the state run $\mathbf{x}_{x,\mathbf u}$ of $\sys$ starting from any initial state $x \in X_a$ under policy $\mathbf u(\cdot)$ satisfies $\mathbf{x}_{x,\mathbf u}(k)\cap L^{-1}(\sigma(q',q''))=\emptyset$~$\forall k\in\N$.
\end{lemma}
Observe that for a triplet $(q,q',q'')$, $q,q',q''\in Q$, the corresponding labels in the automaton is given as $p_1=\sigma(q,q')$ and $p_2=\sigma(q',q'')$. Now by using the labeling function $L$, one gets corresponding regions $X_a=L^{-1}(p_1)$ and $X_b=L^{-1}(p_2)$. Thus, one can simply use Proposition \ref{thm1} to provide the result in Lemma \ref{compo_lemma1}.

Lemma \ref{compo_lemma1} uses control barrier functions along with appropriate choices of stationary control policies $\mathbf u(\cdot)$ for elements in $\mathcal P(\mathcal{A}^c)$ as mentioned in Proposition~\ref{thm1}. However, computation of control barrier functions and the policies for each element of $\mathcal P(\mathcal{A}^c)$ can cause ambiguity while utilizing controllers in closed-loop whenever there are more than one outgoing edges from a state of the automaton. To make this more clear, consider elements $\eta_1=(q_0,q_1,q_2)$ and $\eta_2=(q_0,q_1,q_4)$ from Example 1, where there are two outgoing transitions from state $q_1$ (see Figure \ref{fig_compo_automata}). This results in two different reachability problems, namely, reaching sets $L^{-1}(\sigma(q_1,q_2))$ and $L^{-1}(\sigma(q_1,q_4))$ starting from the same set $L^{-1}(\sigma(q_0,q_1))$. Thus computing different control barrier functions and corresponding controllers in such a scenario is not helpful. To resolve this ambiguity, we simply merge such reachability problems into one reachability problem by replacing the set $X_{b}$ in Lemma~\ref{compo_lemma1} with the union of regions corresponding to the alphabets of all outgoing edges. Thus, we get a common control barrier function and a corresponding controller. This enables us to partition $\mathcal P(\mathcal{A}^c)$ and put the elements sharing a common control barrier function and a corresponding control policy in the same partition set. These sets can be formally defined as:
\begin{align*}
\mu_{(q,q',\Delta(q'))}\hspace{-.2em}:=\hspace{-.2em}\{(q,q',q''&)\in\mathcal P(\mathcal{A}^c)\mid q,q',q''\in Q\text{ and }q''\in\Delta(q')\}.
\end{align*}
The control barrier function and the control policy corresponding to the partition set $\mu_{(q,q',\Delta(q'))}$ are denoted by $B_{\mu_{(q,q',\Delta(q'))}}(x)$ and $\mathbf u_{\mu_{(q,q',\Delta(q'))}}(x)$, respectively. Thus, for all $\eta\in\mathcal P(\mathcal{A}^c)$, we have 
\begin{align}\label{eq_controller1}
B_\eta(x)=B_{\mu_{(q,q',\Delta(q'))}}(x)\text{ and } \mathbf u_\eta(x)=\mathbf u_{\mu_{(q,q',\Delta(q'))}}(x),
\text{ if } \eta\in\mu_{(q,q',\Delta(q'))}.
\end{align}

\subsection{Control Policy}\label{aaaaa1}
From the above discussion, one can readily observe that we have different stationary control policies at different locations of the automaton which can be interpreted as a switching control policy.
Next, we define the automaton representing the switching mechanism for control policies. Consider the DBA $\mathcal{A}^c=(Q,Q_0,\Pi,\delta,F)$ as discussed in Section~\ref{compo_runs}, where $\Delta(q)$ denotes the set of all successor states of $q\in Q$. Now, the switching mechanism is given by an automata $\mathcal{A}_{\mathfrak m}=(Q_{\mathfrak m},Q_{\mathfrak m 0},\Pi_{\mathfrak m},\delta_{\mathfrak m})$, where $Q_{\mathfrak m}:=Q_{\mathfrak m 0}\cup\{(q,q',\Delta(q'))\mid q,q'\in Q\}$ is the set of states, $Q_{\mathfrak m 0}:=\{(q_0,\Delta(q_0))\mid q_0\in Q_0\}$ is the set of initial states, $\Pi_{\mathfrak m}=\Pi$, and the transition relation $(q_{\mathfrak m},\sigma,q_{\mathfrak m}')\in \delta_{\mathfrak m}$ is defined as
\begin{itemize}
	\item for all $q_{\mathfrak m}=(q_0,\Delta(q_0))\in Q_{\mathfrak m 0}$, 
	\vspace{0.1cm}
	\begin{itemize}
		\item[] \!\!$(q_0,\Delta(q_0))\!\!\overset{\sigma(q_0,q'')}{\longrightarrow}\!\!(q_0,q'',\Delta(q''))$, where $q_0\!\!\overset{\sigma(q_0,q'')}{\longrightarrow}q''$;
	\end{itemize}
	\item for all $q_{\mathfrak m}=(q,q',\Delta(q'))\in Q_{\mathfrak m}\setminus Q_{\mathfrak m 0}$,
	\vspace{0.1cm}
	\begin{itemize}
		\item[] \!\!$(q,q',\Delta(q'))\!\!\overset{\sigma(q',q'')} {\longrightarrow}\!\!(q',q'',\Delta(q''))$, such that $q,q',q''\in Q$, $q'\!\!\overset{\sigma(q',q'')}{\longrightarrow}q''$.
	\end{itemize}
\end{itemize}
The control policy that is a candidate for solving Problem~\ref{prob1} is given by 
\begin{equation}\label{eq_policy_compo}
\rho(x,q_{\mathfrak m})=\mathbf u_{\mu_{(q_{\mathfrak m}')}}(x), \quad \forall (q_{\mathfrak m},L(x),q_{\mathfrak m}')\in\delta_{\mathfrak m}.
\end{equation}
\begin{remark}
	The control policy in \eqref{eq_policy_compo} is a policy on the augmented space $X\times Q_{\mathfrak m}$. Such a policy is equivalent to a history dependent policy on the state set $X$ of the system as discussed in the last paragraph of Subsection~\ref{Interconnect_sys1} (see \cite{tkachev2013quantitative} for a proof). 
\end{remark}

\begin{example}(continued)
	Consider DBA $\mathcal A^c$ in Figure~\ref{fig_compo_automata}. Assume we have control barrier functions and corresponding control policies as given in \eqref{eq_controller1}. The automaton $\mathcal{A}_{\mathfrak m}=(Q_{\mathfrak m},Q_{\mathfrak m 0},\Pi_{\mathfrak m},\delta_{\mathfrak m})$ modeling the switching mechanism between control policies is shown in Figure~\ref{fig:switching}. 
	\begin{figure}[t] 
		\centering
		\includegraphics[scale=0.15]{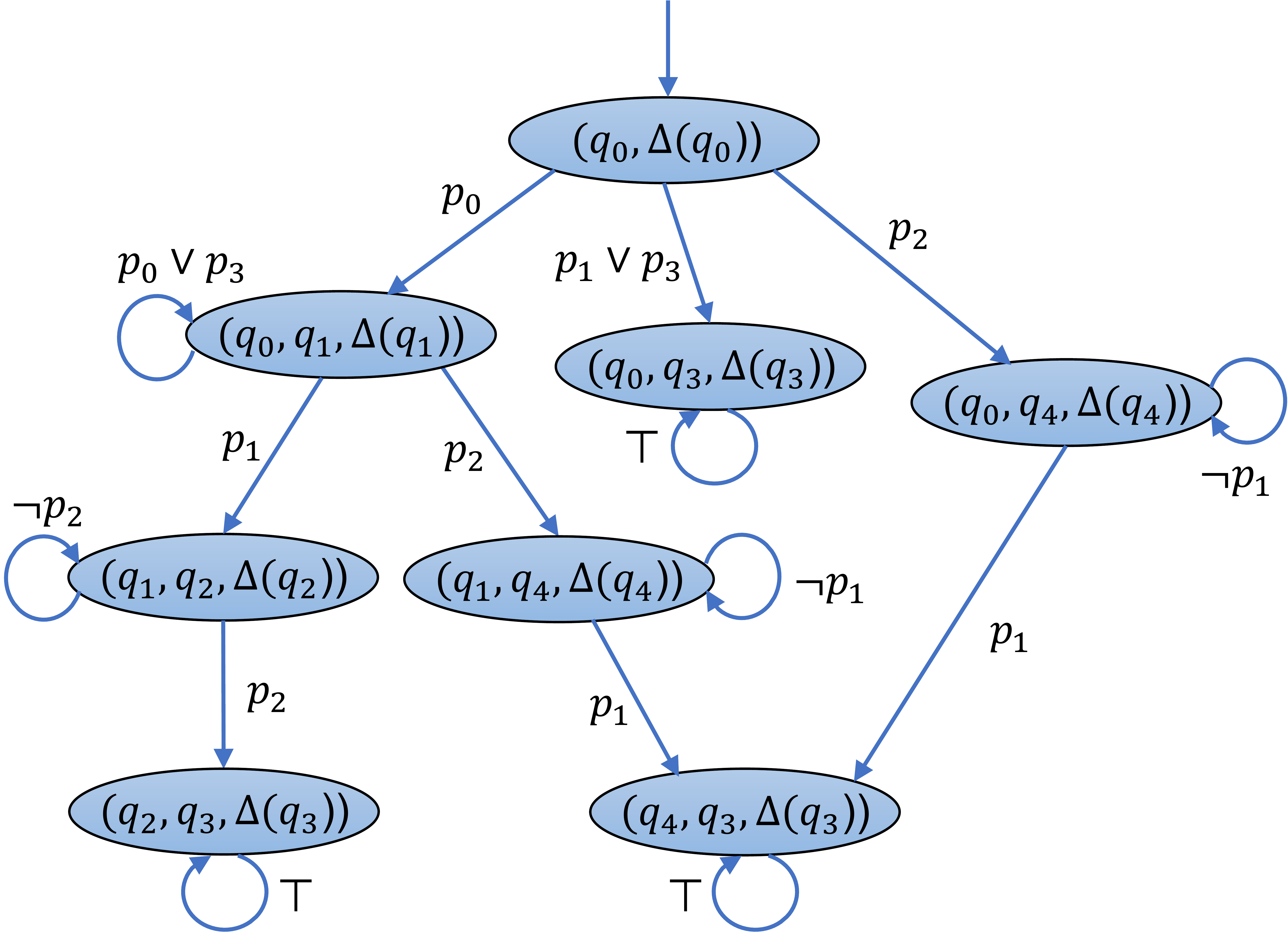} 
		\caption{Automata $\mathcal{A}_{\mathfrak m}$ representing switching mechanism for controllers.}
		\label{fig:switching}
		\vspace{-1.2em}
	\end{figure}
\end{example}

In the next theorem, we show that the policy given in \eqref{eq_policy_compo} is indeed a solution for Problem~\ref{prob1}.
\begin{theorem}\label{thm2}
	Given $p\in\Pi$, assume that there exists $(q,q',q'')\in\mathcal P^p(\overline{\textbf{q}})$, for all $\overline{\textbf{q}}\in\mathcal{R}^p$ for which we have a control barrier function and a controller as given in \eqref{eq_controller1}. Then the state run $\mathbf{x}_{x,\rho}$ of $\sys$ starting from any initial state $x \in L^{-1}(p)$ under policy $\rho$ given in \eqref{eq_policy_compo} satisfies the accepting language of DCA $\mathcal A$, i.e., $L(\mathbf{x}_{x,\rho}(k))\models\mathcal A$ for all $k\in\N$.
\end{theorem}
\begin{proof}
	Consider $p\in\Pi$ and an accepting state run $\textbf{q}\!=\!(q_0^r,q_1^r,\ldots,$ $q_{\mathsf m_r}^r,(q_0^s,q_1^s,\ldots,q_{\mathsf m_s}^s)^{\omega})\in Q^\omega$ in $\mathcal{A}^c$ with $\sigma(q_0^r,q_1^r)=p $. Let the corresponding finite state run be $\overline{\textbf{q}}\in \mathcal{R}^p$ as defined in Subsection \ref{compo_runs}. If for a triplet $(q,q',q'')\in\mathcal P^p(\overline{\textbf{q}})$ one can find a control barrier function with a stationary control policy $\mathbf u(\cdot)$, from Lemma \ref{compo_lemma1} one can conclude $\sigma(\textbf{q})\notin \mathcal{L}(\mathcal{A}^c)$. 
	Now, if there exist control barrier functions and corresponding controllers as defined in \eqref{eq_controller1} for a triplet 
	$(q,q',q'')\in \mathcal P^p(\overline{\textbf{q}})$ for any $\overline{\textbf{q}}\in\mathcal{R}^p$, one has $\sigma(\textbf{q})\notin \mathcal{L}(\mathcal{A}^c)$ for any accepting state run $\textbf{q}= (q_0^r,q_1^r,\ldots,q_{\mathsf m_r}^r$, $(q_0^s,q_1^s,\ldots,q_{\mathsf m_s}^s)^{\omega})\in Q^\omega$ satisfying $\sigma(q_0^r,q_1^r)=p $. By utilizing the definition of labeling function $L$, this implies that the state run $\mathbf{x}_{x,\rho}$ of $\sys$ starting from any initial state $x \in L^{-1}(p)$ under policy $\rho$ given in \eqref{eq_policy_compo} satisfies $L(\mathbf{x}_{x,\rho}(k))\notin \mathcal L(\mathcal A^c)$ for all $k\in\N$. Hence, we have $L(\mathbf{x}_{x,\rho}(k))\in \mathcal L(\mathcal A)$ for all $k\in\N$ and for any initial state $x \in L^{-1}(p)$. This concludes the proof.
\end{proof}
\begin{remark}
	Theorem \ref{thm2} says that in order to satisfy the given specification by the system $\sys$ starting from any initial state $x\in L^{-1}(p)$, one needs to find a control barrier function as in \eqref{eq_controller1} satisfying Lemma \ref{compo_lemma1} for at least one $(q,q',q'')\in\mathcal P^p(\overline{\textbf{q}})$ for each $\overline{\textbf{q}}\in\mathcal{R}^p$. For the rest, one can choose control inputs arbitrarily. 
\end{remark}
\begin{remark}\label{no_barrier}
	For any $(q,q',q'')\in\mu_{(q,q',\Delta(q'))}$, if $L^{-1}(\sigma(q,q'))\cap L^{-1}(\sigma(q',q''))\neq\emptyset$, there exists no control barrier function satisfying conditions in Proposition \ref{thm1}. This follows directly due to the conflict in conditions \eqref{bc1} and \eqref{bc2}. For example consider the triplet $(q_4,q_3,q_3)\in\mathcal{P}^{p_2}(q_0,q_4,q_3)$ in Example 1. There, we have $L^{-1}(p_1)\cap L^{-1}(\top)=L^{-1}(p_1)\neq\emptyset$, so there is no need to search for a control barrier function in this case since there is none.
\end{remark}

A general interpretation of this section can be summarized as follows. Intuitively, control barrier functions are used to provide a guarantee for not reaching an unsafe set starting from an initial set. When dealing with DBA, one should provide control barrier functions ensuring that trajectories are not reaching final states of the automata through all possible paths. To do so, those paths are divided into state runs of length 3 with two atomic propositions associated with it. The regions associated with those two atomic propositions can be treated as sets $X_a$ and $X_b$ in Proposition \ref{thm1} to find such control barrier functions. If we succeed in finding at least one control barrier function in all possible paths, we can provide the result on the overall reachability property (i.e. on reaching final states of DBA).

\section{Compositional Construction of Control Barrier Functions} \label{mainresult}
In this section, we provide a method for compositional construction of control barrier functions for interconnected systems $\sys$ in Definition \ref{interconnectedsystem}. Suppose we are given control subsystems $\sys_i = (X_i,U_i,W_i,f_i,Y_i,h_i)$, $i\in [1,N]$, and assume sets $X_a$ and $X_b$ introduced in Proposition \ref{thm1} can be decomposed as $X_a=\prod_{i=1}^NX_{ai}$ and $X_b=\prod_{i=1}^NX_{bi}$. Note that sets $X_a$ and $X_b$ are associated with some atomic propositions in $\Pi$ through a labeling function $L:X\rightarrow\Pi$. 
This implies that all the sets associated with atomic propositions in $\Pi$ have the decomposed structure as $X_a$ and $X_b$.
The result provided in this section is mainly used to obtain control barrier functions compositionally to satisfy the reachability tasks as given in Lemma \ref{compo_lemma1}. Here, we assume that each control subsystem $\sys_i$ admits a local control barrier function as defined next. 
\begin{definition}\label{lbc}
	Let $\sys_i = (X_i,U_i,W_i,f_i,Y_i,h_i)$ be a control subsystem, where $i\in[1;N]$. A function $ \B_i: X_i \to \R_{\geq0}$ is called a local control barrier function for $\sys_i$ if it satisfies the following conditions:
	\begin{align}
	\B_i(x_i)&\geq \alpha_i(\Vert h_i(x_i)\Vert), \quad\quad\forall x_i\in X_{i},\label{lbc0}\\
	\B_i(x_i)&\leq\overline{\epsilon}_{i}, \quad\quad\quad\quad\quad\quad\ \forall x_i\in X_{ai},\label{lbc1}\\
	\B_i(x_i)&>\underline{\epsilon}_{i},  \quad\quad\quad\quad\quad\quad\ \forall x_i\in X_{bi},\label{lbc2}
	\end{align}
	and $\forall x_i\in X_i$  $\exists$ $u_i\in U_i, \forall w_i\in W_i$ such that
	\begin{align}\label{lbc3}
	\B_i(f_i(x_i,w_i,u_i))&\leq\max\{\kappa_i(\B_i(x_i)),\gamma_{wi}(\Vert w_i\Vert)\},	
	\end{align}
	for some $\alpha_i,\kappa_i,\gamma_{wi}\in \mathcal{K}_{\infty}$ with $\kappa_i\leq \I$, and some $\underline{\epsilon}_{i},\overline{\epsilon}_{i}\in\R_{\geq0}$.
\end{definition}
Local control barrier functions of subsystems are mainly for constructing control barrier functions for the interconnected systems and they are not used directly for verifying any reachability task.
\begin{remark}
	Note that condition $\epsilon_1\leq\epsilon_2$ in Definition \ref{bc} requires implicitly that $X_a\cap X_b=\emptyset$. However, in Definition \ref{lbc} we do not require any condition between $\underline{\epsilon}_i$ and $\overline{\epsilon}_i$ because one may have $X_{ai}\cap X_{bi}\ne\emptyset$ even though $X_a\cap X_b=\emptyset$.
\end{remark}

\begin{remark}
	Note that condition \eqref{lbc3} in Definition \ref{lbc} implies that control input $u_i$ only depends on the state $x_i$ and is independent of internal input $w_i$. This allows us to design (if possible) decentralized control policies which do not require state information of other subsystems. However, if we change the sequence of quantifiers in \eqref{lbc3} to $\forall x_i\in X_i$ $\forall w_i\in W_i$ $\exists$ $u_i\in U_i$, then one obtains distributed control policies which require state informations of neighboring subsystems through internal inputs $w_i$.
\end{remark}

For functions $\kappa_i$, $\alpha_i$, and $\gamma_{wi}$ associated with $\B_i$, $\forall i\in [1;N]$, appeared in Definition \ref{lbc}, we define
\begin{align}\label{gammad}
\!\!\gamma_{ij}\!\Let\!\left\{
\begin{array}{lr}
\!\!\kappa_{i}\quad\quad\quad\quad\, \text{if}\quad i=j,\\
\!\!\gamma_{wi}\circ\alpha_{j}^{-1}\quad\, \text{if} \quad i\neq j, 
\end{array}\right.\forall i,j \in [1;N].
\end{align}
In order to establish the main compositionality results of the paper, we raise the following small-gain type assumption.
\begin{assumption}\label{sg}
	Assume that functions $\gamma_{ij}$ defined in \eqref{gammad} satisfy
	\begin{align}\label{SGC}
	\gamma_{i_1i_2}\circ\gamma_{i_2i_3}\circ\cdots\circ\gamma_{i_{r-1}i_r}\circ\gamma_{i_ri_1}<\mathcal{I}_d,
	\end{align}
	$\forall(i_1,\ldots,i_r)\in\{1,\ldots,N\}^r$, where $r\in \{1,\ldots,N\}$.
	
	Note that by using Theorem 5.2 in \cite{090746483}, the small-gain condition \eqref{SGC} implies that there exist $\varphi_i \in \mathcal{K}_{\infty}$, $\forall i\in [1;N]$, satisfying 
	\begin{align}\label{gam}
	&\max\limits_{j\in [1;N]}\{\varphi^{-1}_i\circ\gamma_{ij}\circ\varphi_j\}<\mathcal{I}_d.
	\end{align}
\end{assumption}
The next theorem provides a compositionality approach to compute a control barrier function for interconnected system $\sys$ in Definition \ref{interconnectedsystem} via local control barrier functions of subsystems $\sys_i$. 
\begin{theorem}\label{thm:3}
	Consider the interconnected control system
	$\sys=\mathcal{I}(\sys_1,\ldots,\sys_N)$ induced by
	$N\in\N_{\ge1}$
	control subsystems~$\sys_i$. Assume that each $\sys_i$ admits a local control barrier function $\B_i$ as defined in Definition \ref{lbc}. Let Assumption \ref{sg} hold and $\max\limits_{i\in[1;N]}\!\{ \varphi^{-1}_{i}(\overline{\epsilon}_{i}) \}\leq\max\limits_{i\in[1;N]}\!\{ \varphi^{-1}_{i}(\underline{\epsilon}_i) \}$. 
	Then, function $ \B:X \to \R_{\geq0}$ defined as
	\begin{align}\notag
	\B(x)\Let\max\limits_{i\in[1;N]}\{ \varphi^{-1}_{i}\circ \B_i(x_i) \}, 
	\end{align}
	is a control barrier function for the interconnected control system $\sys$ satisfying conditions \eqref{bc1} and \eqref{bc2} in Proposition \ref{thm1} with $X_a=\prod_{i=1}^NX_{ai}$ and $X_b=\prod_{i=1}^NX_{bi}$.
\end{theorem}
\begin{proof}
	First, let $\kappa=\max\limits_{i,j\in [1,N]}\{\varphi^{-1}_i\circ\gamma_{ij}\circ\varphi_j\}$. It follows from \eqref{gam} that $\kappa<\mathcal{I}_d$.  
	
	Now $\forall x= [x_1;\dots;x_N]\in \prod_{i=1}^N X_{i}=X$  $\exists u= [u_1;\dots;u_N]\in \prod_{i=1}^N U_{i}=U$ such that one gets the following chain of inequalities 
	\begin{align}\notag
	\B(f(x,u))=&\max\limits_{i}\{\varphi^{-1}_{i}\circ \B_i(f_i(x_i,u_i,w_i))\}\\\notag
	\leq& \max\limits_{i}\Big\{\varphi^{-1}_{i}\big(\max\{\kappa_i( \B_i(x_i)),\gamma_{wi}(\Vert w_i\Vert )\}\big)\Big\}\\\notag
	=& \max\limits_{i}\Big\{\varphi^{-1}_{i}\big(\max\{\kappa_i( \B_i(x_i)),\gamma_{wi}(\max\limits_{j,j\neq i}\{\Vert w_{ij}\Vert \})\}\big)\Big\}\\\notag
	=& \max\limits_{i}\Big\{\varphi^{-1}_{i}\big(\max\{\kappa_i( \B_i(x_i)),\gamma_{wi}(\max\limits_{j,j\neq i}\{\Vert y_{ji}\Vert \})\}\big)\Big\}\\\notag
	=&\max\limits_{i}\Big\{\varphi^{-1}_{i}\big(\max\{\kappa_i( \B_i(x_i)),\gamma_{wi}(\max\limits_{j,j\neq i}\{\Vert h_{ji}(x_j)\Vert\})\}\big)\Big\}\\\notag
	\leq&\max\limits_{i}\Big\{\varphi^{-1}_{i}\big(\max\{\kappa_i( \B_i(x_i)),\gamma_{wi}(\max\limits_{j,j\neq i}\{\Vert h_{j}(x_j)\Vert\})\}\big)\Big\}\\\notag
	\leq&\max\limits_{i}\Big\{\varphi^{-1}_{i}\big( \max\{\kappa_i( \B_i(x_i)),\gamma_{wi}(\max\limits_{j,j\neq i}\{\alpha^{-1}_{j}\circ \B_j(x_j)\}\big)\Big\}\\\notag
	\leq&\max\limits_{i,j}\Big\{\varphi^{-1}_{i}\circ\gamma_{ij}\circ \B_j(x_j)\Big\}\\\notag
	=&\max\limits_{i,j}\Big\{\varphi^{-1}_{i}\circ\gamma_{ij}\circ\varphi_{j}\circ\varphi^{-1}_{j}\circ \B_j(x_j)\Big\}\\\notag
	\leq&\max\limits_{i,j,l}\Big\{\varphi^{-1}_{i}\circ\gamma_{ij}\circ\varphi_{j}\circ\varphi^{-1}_{l}\circ \B_l(x_l)\Big\}\\\notag
	=&\max\limits_{i,j}\Big\{\varphi^{-1}_{i}\circ\gamma_{ij}\circ \varphi_{j}\circ \B(x)\Big\}=\kappa( \B(x)),\notag
	\end{align}
	satisfying condition \eqref{bc3}. \\
	Now, we show that conditions \eqref{bc1} and \eqref{bc2} hold. From conditions \eqref{lbc1} and \eqref{lbc2}, $\forall x= [x_1;\dots;x_N]\in\prod_{i=1}^N X_{ai}=X_a$, one has
	\begin{align*}
	\B(x)&\!=\!\!\max\limits_{i\in[1;N]}\!\{\varphi^{-1}_{i}\!\!\circ \!\B_i(x_i) \}\!\leq\!\!\!\max\limits_{i\in[1;N]}\!\{ \varphi^{-1}_{i}(\overline{\epsilon}_{i}) \},
	\end{align*}
	and $\forall x= [x_1;\dots;x_N]\in \prod_{i=1}^N X_{bi}=X_b$ 
	\begin{align*}
	\B(x)&\!=\!\!\max\limits_{i\in[1;N]}\!\{\varphi^{-1}_{i}\!\!\circ \!\B_i(x_i) \}\!>\!\max\limits_{i\in[1;N]}\!\{ \varphi^{-1}_{i}(\underline{\epsilon}_{i}) \},
	\end{align*}
	satisfying conditions \eqref{bc1} and \eqref{bc2} with $$\epsilon_1=\max\limits_{i\in[1;N]}\{ \varphi^{-1}_{i}(\overline{\epsilon}_i) \},\epsilon_2=\max\limits_{i\in[1;N]}\{ \varphi^{-1}_{i}(\underline{\epsilon}_i) \}.$$ 
	This concludes the proof.
\end{proof}

Now, we provide a discussion about the feasibility of inequality 
\begin{align}\label{ineq}
\max\limits_{i\in[1;N]}\!\{ \varphi^{-1}_{i}(\overline{\epsilon}_{i}) \}\leq\max\limits_{i\in[1;N]}\!\{ \varphi^{-1}_{i}(\underline{\epsilon}_i) \},
\end{align}
required in Theorem \ref{thm:3}. In general,  inequality \eqref{ineq} is not very restrictive. Indeed, functions $\varphi_{i}$ in \eqref{gam} play the role of rescaling the barrier functions of the individual subsystems while normalizing the effect of internal gains of other subsystems (see \cite{090746483} for a similar discussion in the context of Lyapunov stability). Due to this scaling, one can expect that such an inequality holds in many applications.

In the case that $X_{ai}\cap X_{bi}=\emptyset$, $\forall i\in[1;N]$, inequality \eqref{ineq} always holds with $\max\limits_{i\in[1;N]}\!\{\overline{\epsilon}_i \}\leq\min\limits_{i\in[1;N]}\!\{\underline{\epsilon}_i \}$. Note that we can always impose such a condition over $\underline{\epsilon}_i$ and $\overline{\epsilon}_i$ whenever $X_{ai}\cap X_{bi}=\emptyset,\forall i\in[1;N]$.
~In the case where $\varphi_{i}=\varphi_{j},\forall i,j\in[1;N]$, inequality \eqref{ineq} simply reduces to $\max\limits_{i\in[1;N]}\!\{\overline{\epsilon}_i \}\leq\max\limits_{i\in[1;N]}\!\{\underline{\epsilon}_i \}$.

\begin{remark}
	In the context of stability analysis of interconnected nonlinear control systems, condition \eqref{SGC} is commonly used to show different stability proprieties (e.g., uniform asymptotic stability or input-to-state stability) for the entire network by investigating stability criteria for subsystems. Moreover, condition \eqref{SGC} is also been shown to be tight and cannot be weakened in the context of stability verification of interconnected systems. We refer interested readers to \cite{Dashkovskiy2007} for more details on the tightness analysis of  small-gain condition \eqref{SGC}.
\end{remark}

\begin{remark}
	Here, we provide a general guideline on the computation of $\mathcal{K}_{\infty}$ functions $\varphi_i, i\in[1;N]$ as follows: $(i)$ In the case of having $N\ge 1$ subsystems, functions $\varphi_i, i\in[1;N]$, can be constructed numerically using the algorithm proposed in \cite{Eaves} and the technique provided in \cite[Proposition 8.8]{090746483}, see \cite[Chapter 4]{Rufferp}; $(ii)$ Simple construction techniques are provided in \cite{JIANG} and \cite[Section 9]{090746483} for the case of two and three subsystems, respectively; $(iii)$ the $\mathcal{K}_{\infty}$ functions $\varphi_i, i\in[1;N]$, can be always chosen as identity functions provided that $\gamma_{ij}<\mathcal{I}_d$, $\forall~ i,j\in [1;N]$, for functions $\gamma_{ij}$ appeared in \eqref{gammad}.
\end{remark}

\subsection{Computation of Local Control Barrier Functions}
Proving the existence of a control barrier function and finding one are in general hard problems. However, under some assumptions over systems dynamics, control inputs, and labeling functions, one can search for a local control barrier functions and corresponding control policies of specific forms. In this subsection, we provide two potential solutions: one using sum-of-squares (SOS) program and the other one using counterexample guided inductive synthesis (CEGIS). 
\subsubsection{Sum-of-squares program}\label{SOSS}
In order to formulate conditions in Definition~\ref{lbc} as an SOS optimization to search for a polynomial local control barrier function $\B_i(\cdot)$ and a polynomial stationary control policy $\mathbf u_i(\cdot)$, we raise the following assumption.
\begin{assumption}
	\label{ass:BC}
	Subsystem $\sys_i$ has a continuous state set $X_i\subseteq \mathbb R^{n_i}$, a continuous external input set $U_i\subseteq \R^{m_i}$, and a continuous internal input set $W_i\subseteq \R^{p_i}$. Its transition function $f_i:X_i\times U_i\times W_i \rightarrow X_i$ is polynomial in variables $x_i$, $u_i$, and $w_i$.
\end{assumption}
The following lemma provides a set of sufficient conditions for the existence of local control barrier functions required in Theorem~\ref{thm:3}, which can be solved as an SOS optimization.
\begin{lemma}
	\label{sos}
	Suppose Assumption~\ref{ass:BC} holds and sets $X_{ai},X_{bi}, X_i$ can be defined as $X_{ai}=\{x_i\in\R^{n_i}\mid g_{ai}(x_i)\geq0\}$, $X_{bi}=\{x_i\in\R^{n_i}\mid g_{bi}(x_i)\geq0\}$, $X_i=\{x_i\in\R^{n_i}\mid g_i(x_i)\geq0\}$, and $W_i=\{w_i\in\R^{p_i}\mid g_{wi}(w_i)\geq0\}$, where the inequalities are defined element-wise and $g_{ai},g_{bi},g_{i},g_{wi}$ are vectors of polynomial functions. 
	Suppose there exists a sum-of-squares polynomial $\B_i(x_i)$, polynomials $\lambda_{u_{ji}}(x_i)$ corresponding to the $j^{\text{th}}$ input in $u_i=(u_{1i},u_{2i},\ldots,u_{m_i i})\in U_i\subseteq \R^{m_i}$, and vectors of sum-of-squares polynomials $\lambda_{ai}(x_i)$, $\lambda_{bi}(x_i)$, $\lambda_i(x_i)$, $\overline\lambda_i(x_i)$, $\lambda_{wi}(w_i)$ of appropriate size, and $\hat{\alpha}_i,\hat{\kappa}_i,\hat{\gamma}_{wi}\in \mathcal{K}_{\infty}$ with $\hat{\kappa}_i\leq \I$ such that following expressions are sum-of-squares polynomials:
	\begin{align}
	&\B_i(x_i)-\hat{\alpha}_i(\Vert h_i(x_i)\Vert)-\lambda_i^T(x_i)g_i(x_i),\label{eq:sos1}\\
	&-\B_i(x_i)+\overline{\epsilon}_i-\lambda_{ai}^T(x_i)g_{ai}(x_i),\label{eq:sos2}\\
	&\B_i(x_i)-\underline{\epsilon}_i-\lambda_{bi}^T(x_i)g_{bi}(x_i),\label{eq:sos3}\\
	&-\B_i(f_i(x_i,w_i,u_i))+\hat{\kappa}_i(\B_i(x_i))+\hat{\gamma}_{wi}(\Vert w_i\Vert)-\hspace{-0.3em}\sum_{j=1}^{m_i}(u_{ji}-\lambda_{u_{ji}}(x_i))-\overline\lambda_i^T(x_i) g_i(x_i)-\lambda_{wi}^T(w_i) g_{wi}(w_i),\label{eq:sos4}
	\end{align}
	where $\underline{\epsilon}_{i},\overline{\epsilon}_{i}$ are the constants introduced in Definition \ref{lbc}. Then $\B_i(x_i)$ satisfies conditions \eqref{lbc0}-\eqref{lbc3} in Definition \ref{lbc} and $u_i=[\lambda_{u_{1i}}(x_i);\ldots,\\\lambda_{u_{m_ii}}(x_i)]$, $i\in[1,N]$, is the corresponding controller.
\end{lemma}   
\begin{proof}
	Following a similar argument as the one in the proof of Lemma 5.6 in \cite{jagtap2019formal}, conditions \eqref{eq:sos1}-\eqref{eq:sos4} imply
		\begin{align}
		\B_i(x_i)&\geq \hat{\alpha}_i(\Vert h_i(x_i)\Vert), \quad\quad\forall x_i\in X_{i},\label{lbcs0}\\
		\B_i(x_i)&\leq\overline{\epsilon}_{i}, \quad\quad\quad\quad\quad\quad\ \forall x_i\in X_{ai},\label{lbcs1}\\
		\B_i(x_i)&>\underline{\epsilon}_{i},  \quad\quad\quad\quad\quad\quad\ \forall x_i\in X_{bi},\label{lbcs2}
		\end{align}
		and $\forall x_i\in X_i$  $\exists$ $u_i\in U_i, \forall w_i\in W_i$ such that
		\begin{align}\label{lbcs3}
		\B_i(f_i(x_i,w_i,u_i))&\leq\hat{\kappa}_i(\B_i(x_i))+\hat{\gamma}_{wi}(\Vert w_i\Vert).	
		\end{align}
		By using Theorem 1 in \cite{arxiv}, condition \eqref{lbcs3} can be written as 
		\begin{align*}
		\B_i(f_i(x_i,w_i,u_i))&\leq\max\{\kappa_i(\B_i(x_i)),\gamma_{wi}(\Vert w_i\Vert)\},	
		\end{align*}
		where $\kappa_i=\mathcal{I}_d-(\mathcal{I}_d-\psi_i)\circ(\mathcal{I}_d-\hat{\kappa}_i)$, $\gamma_{wi}=(\mathcal{I}_d-\hat{\kappa}_i)^{-1}\circ\psi_i^{-1}\circ\hat{\gamma}_{wi}$,
		with $\psi_i\in\mathcal{K}_{\infty}$ chosen arbitrarily such that $\psi_i<\mathcal{I}_d$. Let $\alpha_i=\hat{\alpha}_i$ and this concludes the proof.
\end{proof} 
\begin{remark}
	Note that function $\hat\kappa_i(\cdot)$ in \eqref{eq:sos4} can cause nonlinearity on the unknown parameters of $\mathcal{B}_i$. A possible way to avoid this is to consider a linear function $\hat\kappa_i(r)=c_ir, \forall r\geq0 $, with some constant $0<c_i<1$. Then one can use bisection method to minimize the value of $c_i$.
\end{remark}
One can utilize existing tools such as \texttt{SOSTOOL} \cite{sostools} in conjunction with a semidefinite programming solver such as \texttt{SeDuMi} \cite{sturm1999using} to compute a sum-of-squares polynomial $\B_i(x_i)$ satisfying \eqref{eq:sos1}-\eqref{eq:sos4}.  

\subsubsection{Counter-example guided synthesis approach}\label{cegis}
This approach uses feasibility solvers for finding local control barrier functions of a given parametric form using Satisfiability Modulo Theories (SMT) solvers such as \texttt{Z3} \cite{z3}, \texttt{MathSAT} \cite{cimatti2013mathsat5}, or \texttt{dReal} \cite{Gao.2013}. In order to use the CEGIS framework, we raise the following assumption.
\begin{assumption}\label{assum1}
	Each control subsystem $\sys_i$, $i \in [1;N]$, has compact state set $X_i$, compact internal input set $W_i$, and a finite input set $U_i$.
\end{assumption}
Under Assumption \ref{assum1}, conditions \eqref{lbc0}-\eqref{lbc3} can be rephrased as a satisfiability problem which can be searched for parametric local control barrier function using the CEGIS approach. The feasibility condition that is required to be satisfied for the existence of a local control barrier function $\B_i$ is given in the following lemma.
\begin{lemma} \label{cbclem}
	Consider control subsystem $\sys_i \!=\! (X_i,U_i,W_i,f_i,Y_i,h_i)$ satisfying Assumption \ref{assum1}.
	Suppose there exists a function $\B_i(x_i)$ and $\mathcal{K}_{\infty}$ functions $\hat{\alpha}_i, \hat{\kappa}_i$ and $\hat{\gamma}_{wi}$ such that the following expression is true
	\begin{align}\notag
	&\bigwedge_{x_i \in X_i}\!\! \B_i(x_i) \!\geq\! \hat{\alpha}_i(\|h_i(x_i)\|) \bigwedge_{x_i \in X_{ai}} \!\!\B_i(x_i) \!\leq\! \overline{\epsilon}_i 
	\bigwedge_{x_i \in X_{bi}} \!\!\B_i(x_i) >\underline{\epsilon}_i \\\label{feas1}  
	&\bigwedge_{x_i \in X_i}\!\!\Big(\!\!\bigvee_{u_i\in U_i} \!\!\Big(\!\!\bigwedge_{w_i \in W_i}\!\!\!\!\big(\B_i(f_i(x_i,w_i,u_i))\leq\hat{\kappa}_i(\B_i(x_i))+\hat{\gamma}_{wi}(\Vert w_i\Vert)\big)\Big)\Big),
	\end{align}
	where $\underline{\epsilon}_{i},\overline{\epsilon}_{i}$ are the constants introduced in Definition \ref{lbc}. Then $\B_i(x_i)$ satisfies conditions \eqref{lbc0}-\eqref{lbc3} in Definition \ref{lbc}.
\end{lemma}
Note that condition \eqref{feas1} implies conditions \eqref{lbcs0}-\eqref{lbcs3} which imply \eqref{lbc0}-\eqref{lbc3}.
One can utilize the CEGIS approach to search for parametric barrier functions solving the feasibility problem in \eqref{feas1}. For the detailed discussion on CEGIS approach, we kindly refer interested readers to \cite[Subsection 5.3.2]{jagtap2019formal}.

\section{Case Studies}
\subsection{Room Temperature Control} 
The evolution of the temperature $\mathbf{T}$ of $N$ rooms are described by the interconnected discrete-time model:
\begin{align*}
\sys:
\mathbf{T}(k+1)=A\mathbf{T}(k)+ \alpha_e T_{E}+\alpha_h T_{h}\nu(k),
\end{align*}
where $A\in\R^{N\times N}$ is a matrix with elements $\{A\}_{ii}=(1-2\alpha-\alpha_e-\alpha_h \nu_{i}(k))$,
$\{A\}_{i(i+1)}=\{A\}_{(i+1)i}=\{A\}_{1N}=\{A\}_{N1}=\alpha$, $\forall i\in [1;N-1]$, and all other elements are identically zero, $\mathbf{T}(k)=[\mathbf{T}_1(k);\ldots;\mathbf{T}_N(k)]$,  $\nu(k)=[\nu_1(k);\ldots;\nu_N(k)]$, $T_E=[T_{e1};\ldots;T_{eN}]$, where $\nu_i(k)\in[0,1]$ for all $i\in[1;N]$ represents ratio of the heater valve being open. The other parameters are as follow: $\forall i\in[1;N]$, $T_{ei}=15\,^\circ C$ is the external temperature and $T_h\!=\!55\,^\circ C$ is the heater temperature. Parameters $\alpha\!=\!5\times10^{-2}$, $\alpha_e\!=\!8\times10^{-3}$, and $\alpha_h\!=\!3.6\times10^{-3}$ are heat exchange coefficients. All the parameters are adopted from \cite{jagtap2017quest}.\\
The state set of the system is $T\subseteq \R^N$. We consider regions of interest $X_0=[20.5, 22.5]^N$,  $X_1=[0,20]^N$, $X_2=[23,45]^N$, and $X_3=T\setminus(X_0\cup X_1\cup X_2)$. The set of atomic propositions is given by $\Pi=\{p_0,p_1,p_2,p_3\}$ with labeling function $L(x_j)=p_j$ for all $x_j\in X_j$, $j\in\{0,1,2,3\}$. The objective is to compute a control policy ensuring satisfaction of the specification given by the accepting language of the DCA $\mathcal{A}$ in Figure \ref{fig:temp_automata}. In English, language of $\mathcal A$ entails that if we start in $X_0$ it will always stay away from $X_1$ or $X_2$. Note that, the corresponding DBA $\mathcal{A}^c$ accepting complement of $\mathcal{L}(\mathcal{A})$ has exactly the same structure as in Figure \ref{fig:temp_automata}, but with the B{\"u}chi accepting condition. One can readily see that, we have sets $\mathcal{P}^{p_0}=\{(q_0,q_1,q_2),(q_1,q_2,q_2)\}$ and $\mathcal{P}^{p_1}=\mathcal{P}^{p_2}=\mathcal{P}^{p_3}=\{(q_0,q_2,q_2)\}$. Following Remark \ref{no_barrier}, we only need to compute a control barrier function corresponding to triplet $(q_0,q_1,q_2)$.

In order to apply our compositionality result, we need to decompose the system $\sys$ into subsystems $\sys_i$, $i\in[1;N]$.
Accordingly, by introducing $\sys_i$ described by
\begin{align*}
\sys_i:\left\{
\begin{array}{rl}
\mathbf{T}_i(k+1)\!\!\!\!&=a\mathbf{T}_i(k)+d\omega_i(k)+\alpha_e T_{ei} +\alpha_h T_h \nu_i(k),\\
\mathbf{y}_{i}(k)\!\!\!\!&=\mathbf{T}_i(k),%
\end{array}\right.
\end{align*}
one can readily verify that $\sys=\mathcal{I}(\sys_1,\ldots,\sys_N)$, where $a=1-2\alpha-\alpha_e-\alpha_h \nu_{i}(k) $, $d=[\alpha;\alpha]^T$, and $\omega_i(k)=[\mathbf{y}_{i-1}(k);\mathbf{y}_{i+1}(k)]$ (with $\mathbf{y}_{0}=\mathbf{y}_{N}$ and $\mathbf{y}_{N+1}=\mathbf{y}_{1}$).

To compute local control barrier functions, we solve sum-of-squares program using \texttt{SOSTOOLS} and \texttt{SeDuMi} as described in Subsection~\ref{SOSS}. By using Lemma~\ref{sos}, for all $i\in[1;N]$, we compute local control barrier functions of order 2 as $\B_i(x_i)=0.07456x_i^2-3.18x_i+73.79$ and the corresponding stationary control policy of order 1 as
$\mathbf u_i(x_i)=-0.002398x_i+0.5357$ with $X_{ai}=[20.5,22.5]$, $X_{bi}=[0,20]\cup[23,45]$, constants $\overline{\epsilon}_{i}=\underline\epsilon_{i}=40$, and functions $\hat{\alpha}_i(r)=1.5r$, $\hat{\kappa}_i(r)=0.65r$, and $\hat{\gamma}_{wi}(r)=0.5r$ $\forall r\in \R_{\geq0}$. One can readily verify that the small-gain assumption in \eqref{SGC} holds with $\gamma_{ij}(r)=0.95r$, $\forall r\in \R_{\geq0}$. Then by utilizing results in Theorem \ref{thm:3}, we get overall control barrier function $\B(x)\Let\max_{i\in[1;N]}\{\varphi_{i}^{-1}\circ\B_i(x_i) \} $ with $\varphi_{i}=\mathcal{I}_d$ and corresponding control policy is given by $\mathbf u(x)=[\mathbf u_1(x_1);\ldots;\mathbf u_N(x_N)]$. One can readily see that only one stationary control policy is enough for enforcing the specification, thus we do not need switching mechanism.   
Figure \ref{Compo_fig_resp1} shows the maximum and minimum of state trajectories at each time-step of the closed-loop system $\sys$ with 10000 rooms starting from an initial state in $X_0$.

\begin{figure}[t] 
	\centering
	\includegraphics[scale=0.16]{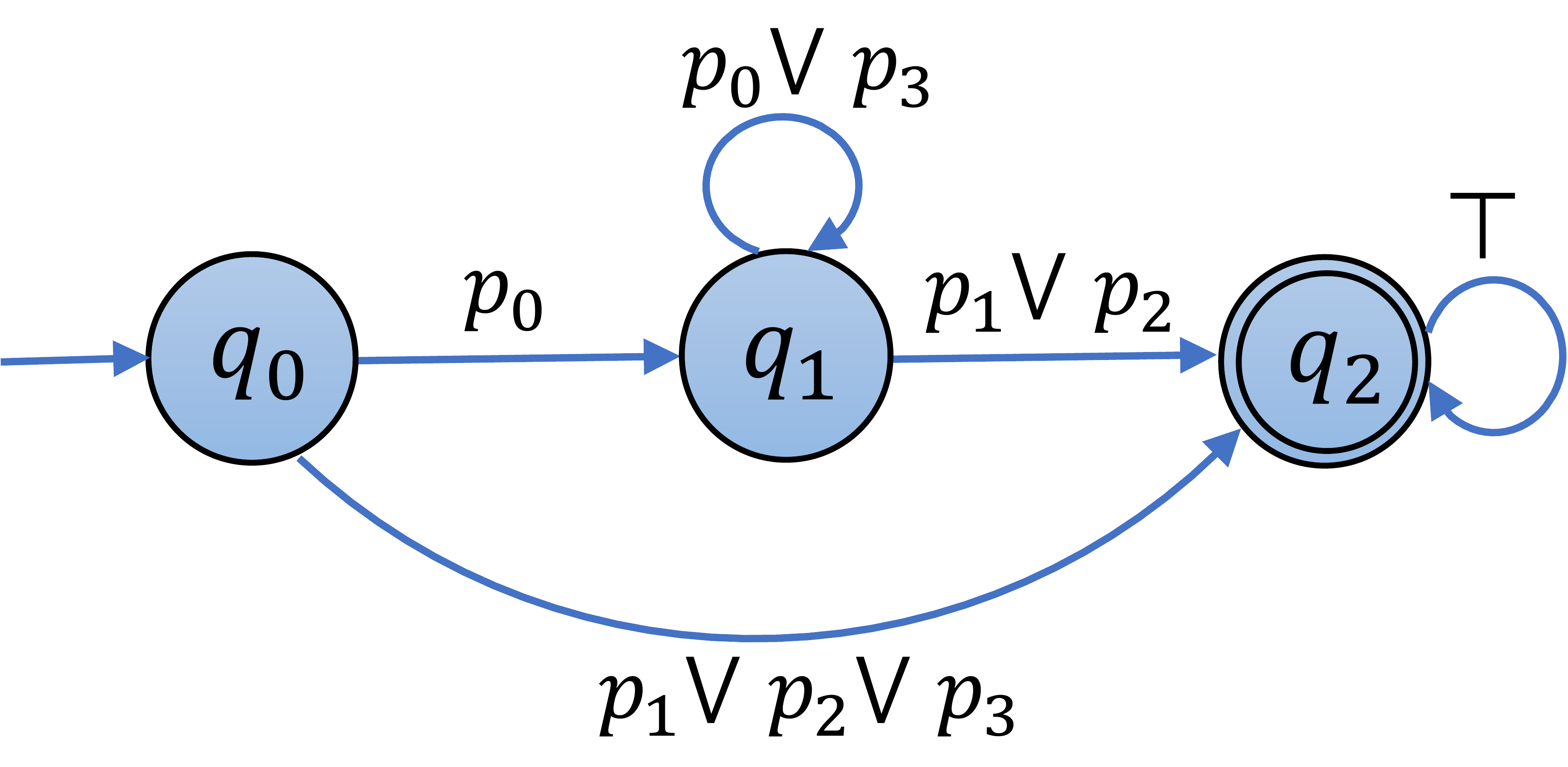}
	\caption{DCA $\mathcal A$ representing specification.}
	\label{fig:temp_automata}
	\vspace{-1em}
\end{figure}
\begin{figure}[t] 
	\vspace{-0.2cm}
	\centering
	\includegraphics[scale=0.42]{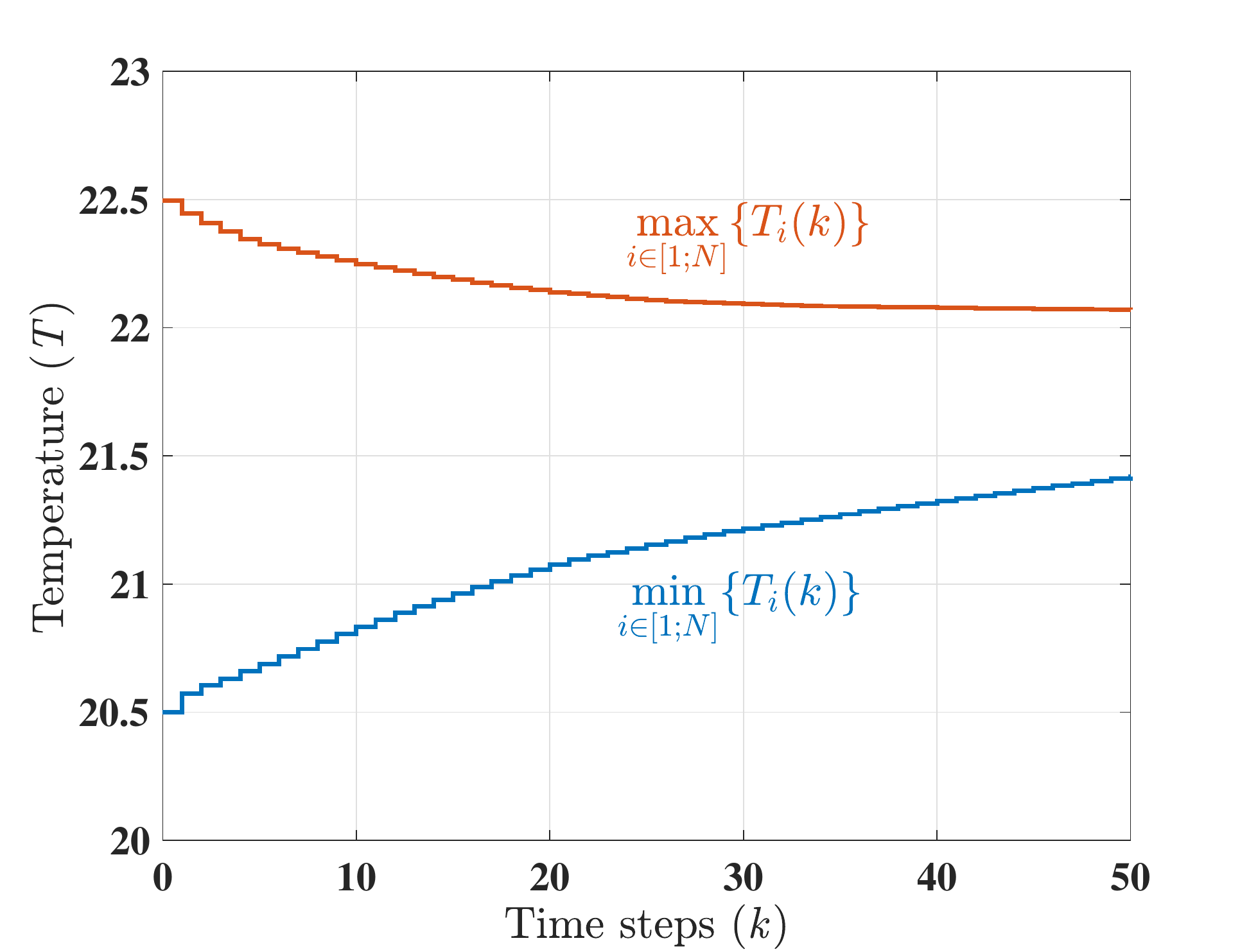}
	\caption{Bounds inside which trajectories are evolving.}
	\label{Compo_fig_resp1}
	\vspace{-1.5em}
\end{figure}
\subsection{Controlled Kuramoto Oscillators}
For the second case study, we consider the Kuramoto oscillator which has large applications in neural networks \cite{ermentrout}, pacemakers
in heart \cite{yongqiang}, automated vehicle coordination \cite{daniel}, and power grids \cite{florian}. In particular, we apply our approach to a variant of the controlled Kuramoto model from \cite{skardal2015control}. The dynamic for an interconnection of $N$-oscillators is given by:
\begin{align*}
\sys:
\theta(k+1)=\theta(k)+ \tau\Omega+\frac{\tau K}{N}\phi(\theta(k))+\nu(k),
\end{align*}
where $\theta(k)\!=\![\!\theta_1(k)\!;\!\dots\!;\!\theta_N(k)\!]\!\in\!\Theta \subseteq \![0,2\pi]^N$ is the phase of the oscillators, $\Omega\!=\![\!\Omega_1\!;\!\dots\!;\!\Omega_N\!]\!=\!\textbf{1}_N$ is the natural frequency of the oscillators,  $\phi(\theta(k))\!=\![\!\sum\limits_{j\in[1;N]}\!\sin(\theta_j(k)\!-\!\theta_1(k))\!;\!\dots\!;\!\sum\limits_{j\in[1;N]}\!\sin(\theta_j(k)\!-\!\theta_N(k))\!]$,
$K=1 $ is the coupling strength, $\tau=0.2$, and control input $\nu(k)=[\nu_1(k);\dots;\nu_N(k)]$, where $\nu_i(k)\in U_i=\{-0.6$, $-0.5, \ldots, 0.5$, $0.6 \}$, $i\in[1;N]$.
We consider regions of interest $X_0=[0, \frac{\pi}{3}]^N$,  $X_1=[\frac{5\pi}{12},\frac{7\pi}{12}]^N$, $X_2=[\frac{2\pi}{3},\pi]^N$, $X_3=[\pi,\frac{4\pi}{3}]^N$, $X_4=[\frac{17\pi}{12},\frac{19\pi}{12}]^N$ and $X_5=[\frac{5\pi}{3},2\pi]^N$, $X_6=X\setminus(X_0\cup X_1\cup X_2\cup X_3\cup X_4\cup X_5)$. The set of atomic propositions is given by $\Pi=\{p_0,p_1,p_2,p_3,p_4,p_5,p_6\}$ with labeling function $L(x_i)=p_i$ for all $x_i\in X_i$, $i\in\{0,1,2,3,4,5,6\}$. The objective is to compute a control policy ensuring satisfaction of the specification given by the accepting language of the DCA $\mathcal{A}$ in Figure \ref{kuramoto_automata}. This corresponds to the LTL specification $(p_1\wedge \square \neg(p_0\vee p_2))\vee (p_4\wedge \square \neg(p_3\vee p_5))$. In English, language of $\mathcal A$ entails that if we start in $X_1$, it will always stay away from $X_0$ or $X_2$ or if we start in $X_4$, it will always stay away from $X_3$ or $X_5$. Note that, the DBA $\mathcal{A}^c$ accepting complement of $\mathcal{L}(\mathcal{A})$ has exactly the same structure as in Figure \ref{kuramoto_automata}, but with the B{\"u}chi accepting condition. As described in Section \ref{compo_runs}, we have sets  $\mathcal{P}^{p_1}=\{(q_0,q_1,q_3),(q_1,q_3,q_3)\}$, $\mathcal{P}^{p_4}=\{(q_0,q_2,q_3),(q_2,q_3,q_3)\}$, and $\mathcal{P}^{p_0}=\mathcal{P}^{p_2}=\mathcal{P}^{p_3}=\mathcal{P}^{p_5}=\mathcal{P}^{p_6}=\{(q_0,q_3,q_3)\}$. Following Remark \ref{no_barrier}, there exists no barrier function corresponding to $(q_0,q_3,q_3)$, $(q_1,q_3,q_3)$, and $(q_2,q_3,q_3)$. This implies that we need to compute only two control barrier functions.
\begin{figure}[t] 
	\centering
	\includegraphics[scale=0.15]{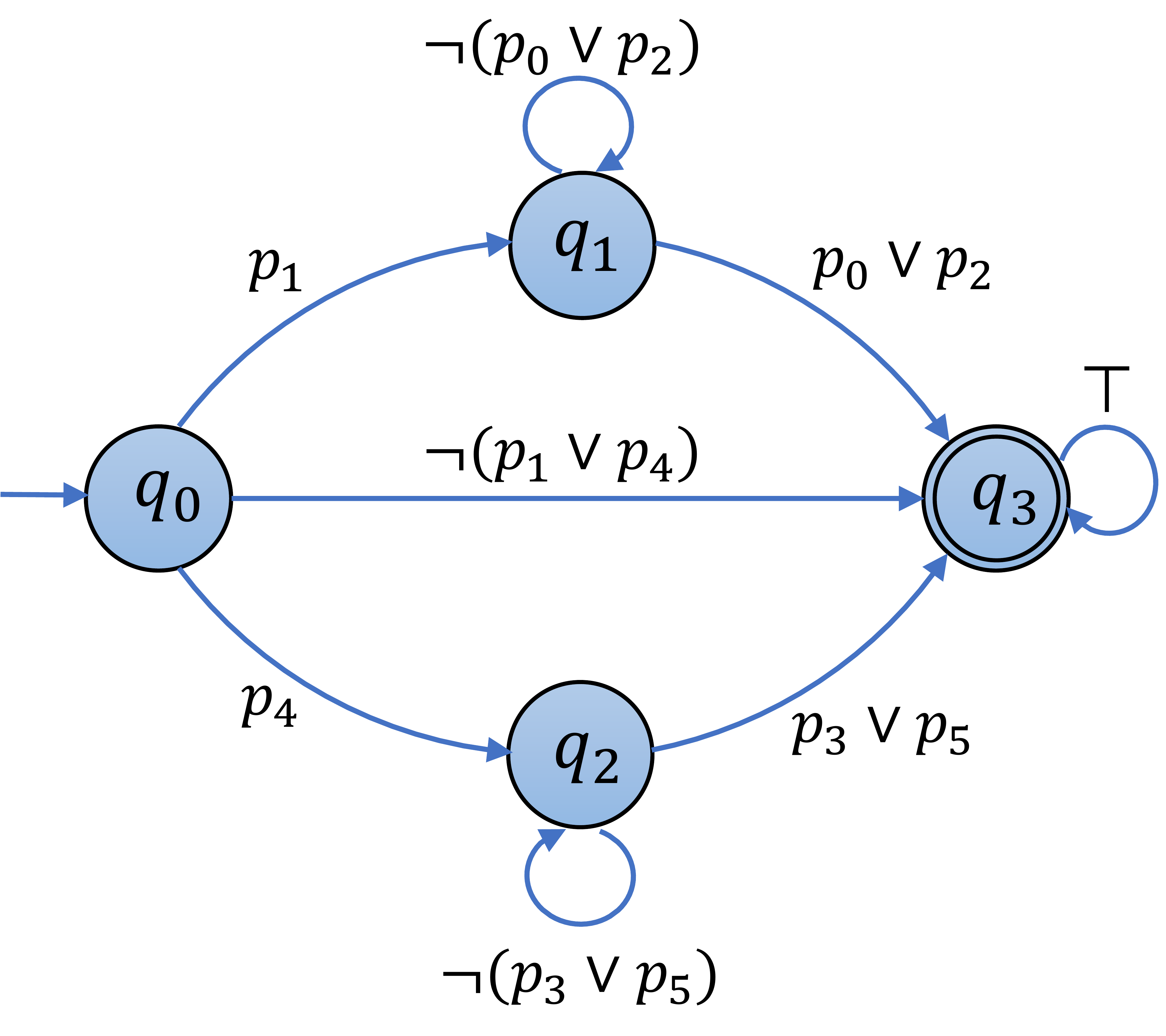}
	\caption{DCA $\mathcal A$ representing the specification.}
	\label{kuramoto_automata}
\end{figure}
Now by introducing subsystems $\sys_i$, $i\in[1;N]$, described by
\begin{align*}
\sys_i:\left\{
\begin{array}{rl}
\!\!\!\theta_i(k+1)=&\!\!\!\!\theta_i(k)\!+\!\tau\Omega_i\!+\!\frac{K\tau}{N}\sum_{j=1}^N\sin(\omega_{ij}(k)\!-\!\theta_i(k))+\nu_i(k),\\
\!\!\!\mathbf{y}_{i}(k)=&\!\!\!\!\theta_i(k),%
\end{array}\right.
\end{align*}
one can readily verify that $\sys=\mathcal{I}(\sys_1,\ldots,\sys_N)$, where $\omega_{ij}=\theta_j$.

To compute these control barrier functions, we apply our compositionality technique and utilize CEGIS approach, as discussed in Subsection \ref{cegis}. For the demonstration of the results, we fix N=10000. The order 2 polynomial local control barrier function corresponding to triplet $(q_0,q_1,q_3)$ is obtained for all $i\in[1;N]$, as $\B_i(x_i)=10.9427x_i^2-34.3775x_i+29$ with $X_{ai}=[\frac{5\pi}{12},\frac{7\pi}{12}]$, $X_{bi}=[0,\frac{\pi}{3}]\cup[\frac{2\pi}{3},\pi]$, constants $\overline{\epsilon}_i=\underline{\epsilon}_i=5$, functions $\hat{\alpha}_i(r)= 0.5 r^2$, $\hat\kappa_i(r)=1.6\times10^{-6} r$, and $\hat\gamma_{w_i}(r)=0.4368r^2$ $\forall r\in\R^+_0$ satisfying conditions in Lemma \ref{cbclem}. Then, by utilizing results in Theorem \ref{thm:3}, we get the overall control barrier function as $\B(x)\Let\max_{i\in[1;N]}\{\varphi_{i}^{-1}\B_i(x_i) \} $ with $\varphi_{i}=\mathcal{I}_d$ and the corresponding \emph{determinized} controller for each subsystem is given by $\mathbf{u}_i(x_i)=\min\{u_i\in U_i\mid\B_i(f_i(x_i,w^*_i,u_i))\leq\hat{\kappa}_i(\B_i(x_i))+\hat{\gamma}_{wi}(\Vert w^*_i\Vert)\}$ for an arbitrarily chosen $w^*_i\in W_i=[0,2\pi]^{N-1}$. Similarly, the local control barrier function corresponding to triplet $(q_0,q_2,q_3)$ is obtained for all $i\in[1;N]$, as $B_i(x_i)=7.2951x_i^2-68.7549x_i+175$ with $X_{ai}=[\frac{17\pi}{12},\frac{19\pi}{12}]$, $X_{bi}=[\pi,\frac{4\pi}{3}]\cup[\frac{5\pi}{3},2\pi]$, constants $\overline{\epsilon}_i=\underline{\epsilon}_i=15$, functions $\alpha_i(r)= 0.5 r^2$, $\hat\kappa_i(r)=1.6\times10^{-6} r$, and $\hat\gamma_{w_i}(r)=0.2912r^2$ for all $r\in\R^+_0$ satisfying conditions in Lemma \ref{cbclem}. The corresponding determinized controller here is also given as $\mathbf{u}_i(x_i)=\min\{u_i\in U_i\mid\B_i(f_i(x_i,w^*_i,u_i))\leq\hat{\kappa}_i(\B_i(x_i))+\hat{\gamma}_{wi}(\Vert w^*_i\Vert)\}$ for an arbitrarily chosen $w^*_i\in W_i$. Note that in both scenarios the small-gain condition in \eqref{SGC} holds with $\gamma_{ij}(s)=0.5824r$ and $\gamma_{ij}(s)=0.8736r$, $\forall r\in\R_{\geq0}$, respectively.
The switching mechanism for controllers to obtain hybrid control policy $\rho(x, q_{\mathfrak m})$ as defined in \eqref{eq_policy_compo} is shown in Figure \ref{swtch1}. Figure \ref{Compo_response2}(a) and Figure \ref{Compo_response2}(b) show the maximum and minimum bounds inside which all the state trajectories of the closed-loop system $\sys$ starting from an initial state in $X_1$ and $X_4$ evolves, respectively. From Figure \ref{Compo_response2}, one can readily check the satisfaction of the given specification. 

\begin{figure}[t] 
	\centering
	\includegraphics[scale=0.15]{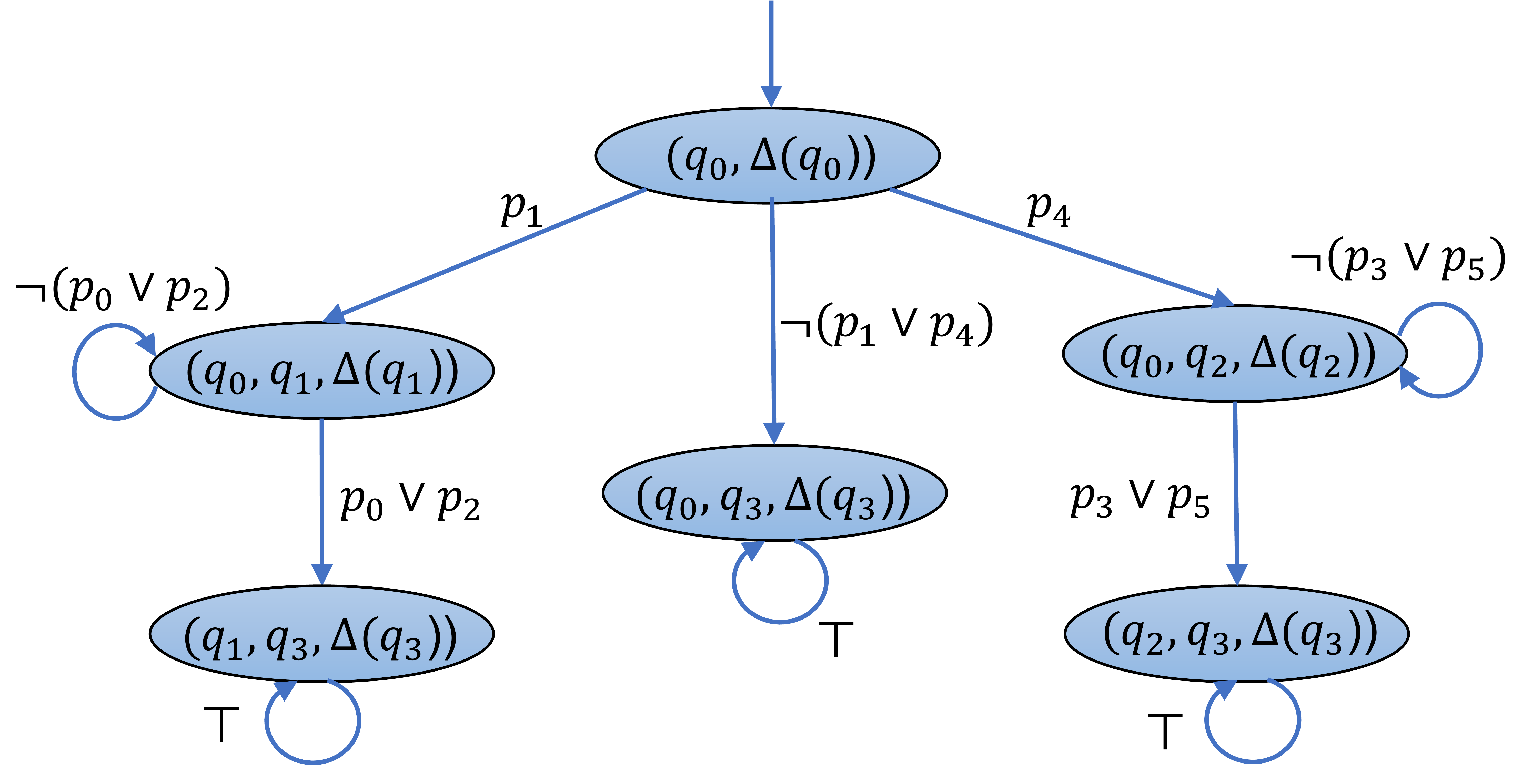}
	\caption{Switching mechanism for controllers.}
	\label{swtch1}
	\vspace{-1em}
\end{figure}
\begin{figure}[t] 
	\centering
	\subfigure[]{\includegraphics[scale=0.40]{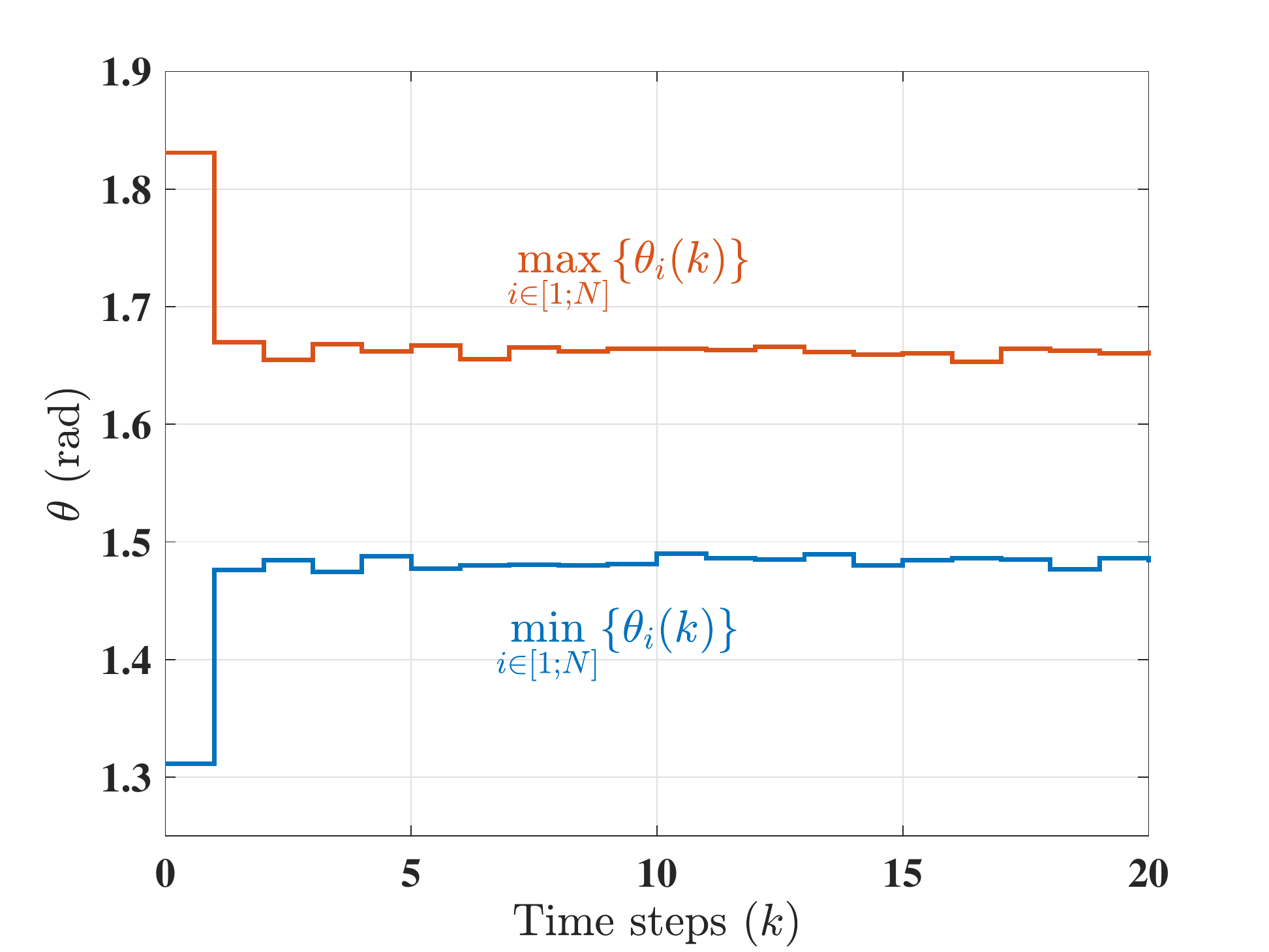}}
	\subfigure[]{\includegraphics[scale=0.40]{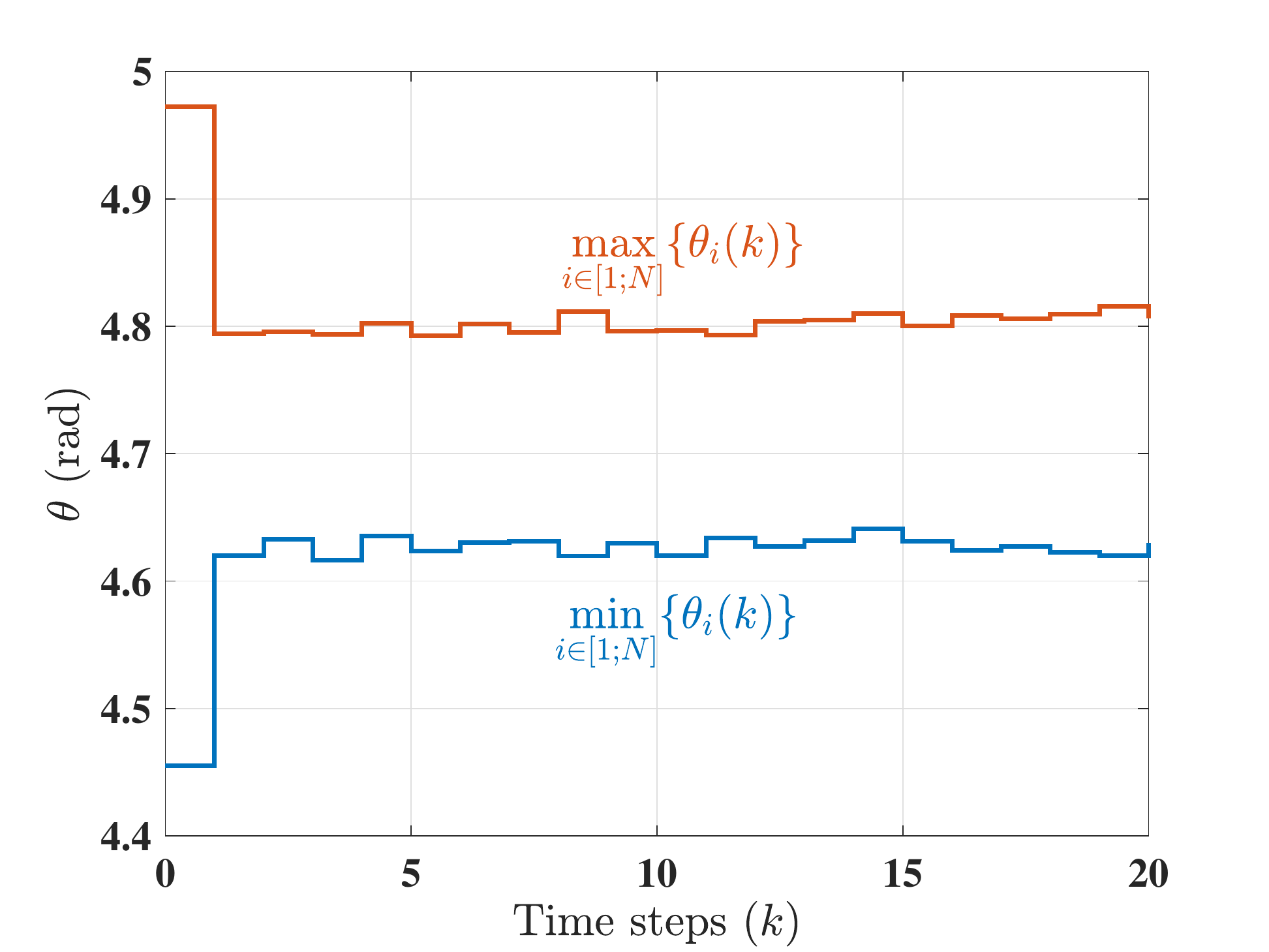}}
	\caption{Bounds inside which trajectories of the Kuramoto model with 10000 oscillators evolve with (a) an initial state starting in region $X_1$ (b) an initial state starting in region $X_4$.}
	\label{Compo_response2}
\end{figure}
\section{Conclusion}
In this work, we proposed a scheme for designing hybrid control policies for interconnected discrete-time control systems enforcing specifications expressed by deterministic co-B\"uchi automata. We first construct automata whose accepting languages are complements of the languages of the original co-B\"uchi automata. Then, we decompose the resulted specification, which is the complement of the original specification, to simpler reachability tasks, then provide a systematic technique to solve these simpler tasks by computing corresponding control barrier functions. We showed that such control barrier functions can be computed compositionally by utilizing a small-gain type reasoning and composing so-called local control barrier functions computed for subsystem. Moreover, we provided two systematic approaches to find local control barrier functions for subsystems based on the sum-of-squares optimization and counter-example guided inductive synthesis approach.

\bibliographystyle{alpha}
\bibliography{biblio}

\end{document}